\documentclass[
reprint,amsfonts, amssymb, amsmath,  showkeys, superscriptaddress,longbibliography,nofootinbib]{revtex4-2}

\usepackage{graphicx}
\usepackage{array}
\usepackage{xcolor}
\usepackage{float}
\usepackage{physics}
\usepackage{enumerate}
\usepackage{enumitem}
\usepackage{lipsum}
\usepackage{amsfonts}
\usepackage{mathtools}
\usepackage[normalem]{ulem} 
\usepackage{dsfont}
\usepackage{bbm}

\usepackage{booktabs}
\usepackage{url}

\usepackage{braket}

\usepackage{amsthm}
\usepackage{bm}% bold math
\usepackage[T1]{fontenc}

\definecolor{C1}{RGB}{52, 89, 149}
\definecolor{C4}{RGB}{202, 21, 81}
\definecolor{C5}{RGB}{180,77,155}
\definecolor{C6}{RGB}{30,125,125}
\usepackage[colorlinks=true,citecolor=C5,linkcolor=C1,urlcolor=C5]{hyperref}

\usepackage{thm-restate}

\newcommand{\img}[2][2.2ex]{%
  \mathrel{\vcenter{\hbox{\includegraphics[height=#1]{#2}}}}%
}

\usepackage{tikz}
\usetikzlibrary{fadings}
\usepackage{mathrsfs}
\usepackage{amsmath}

\newcommand{\sketbra}[2]{|#1 \rangle\! \langle  #2 |}

\usepackage{comment}

\theoremstyle{remark}

\newcommand*{\ot}{\otimes}

\newcommand*{\dg}{\dagger}
\DeclareMathOperator*{\expect}{\mathbb{E}}

\definecolor{dpink}{rgb}{0.85, 0.10, 0.45}

\newcommand{\mcc}{\mathcal{C}}

\newcommand{\mcb}{\mathcal{B}}
\newcommand{\mch}{\mathcal{H}}
\newcommand{\mco}{\mathcal{O}}

\newcommand{\mcp}{\mathcal{P}}
\newcommand{\mcv}{\mathcal{V}}

\newcommand{\mct}{\mathcal{T}}
\newcommand{\mcd}{\mathcal{D}}
\newcommand{\mcw}{\mathcal{W}}

\newcommand{\mbc}{\mathbb{C}}

\newcommand{\mbf}{\mathbb{F}}

\newcommand{\mbs}{\mathbb{S}}

\newcommand{\mbu}{\mathbb{U}}

\newcommand{\mg}{\mathsf{MG}(n)}
\newcommand{\mgb}{\mathsf{MG}(n_B)}

\newcommand{\mbsp}{\mathbb{SP}}

\newcommand{\mbso}{\mathbb{SO}}
\newcommand{\id}{\mathds{1}}

\newcommand{\lm}{\lambda}

\def\tn{^{\otimes n}}

\newcommand{\Lm}{\Lambda }

\def\sg{\sigma}

\def\ad{^{\dagger}}

\begin{document}

\title{Free fermionic black holes as polarised mirrors } 

\author{Maureen Krumt{\"u}nger}
\email{maureen.krumtuenger@student.unimelb.edu.au}
\affiliation{School of Physics, University of Melbourne, Parkville, VIC 3010, Australia}
\author{Maxwell West}
\email{westm@lanl.gov}
\affiliation{Theoretical Division, Los Alamos National Laboratory, Los Alamos, New Mexico 87545, USA}
\affiliation{Quantum Science Center, Oak Ridge, TN 37931, USA}
\affiliation{School of Physics, University of Melbourne, Parkville, VIC 3010, Australia}

\begin{abstract}
Understanding the speed with which various classes of quantum systems can scramble information is a key goal of quantum information theory, with relevance to both fundamental physics and  the study of the efficiency of quantum algorithms. Famously, such understanding enabled  Hayden and Preskill to  show that, under the assumption that black hole dynamics may be modelled by a Haar-random unitary, old black holes act as ``mirrors'',  reflecting infalling quantum information back out again effectively instantaneously. In this work, we investigate which aspects of this mirror behaviour survive beyond the Haar-random setting by considering a random free fermionic model of black hole dynamics. We find that the rate of information release is governed by the underlying representation-theoretic structure; such a black hole acts as a ``polarised mirror'', retaining information encoded in low fermionic-degree irreps while reflecting information encoded in high fermionic-degree irreps much more rapidly. As a result of the retention of low-degree information, we show that faithful recovery of the infalling quantum information is impossible until all but a vanishing fraction of the black hole has evaporated. 
However, augmenting the dynamics with a fixed Clifford unitary can restore their reflective behaviour, even though the ensemble’s frame potential remains unchanged and far above its Haar value. Remarkably, we find that such black holes can outperform their fully Haar-random counterparts: in the high-fidelity limit, recovering a single-qubit message requires asymptotically approximately $37\%$ fewer emitted qubits.
\end{abstract}

\maketitle

\tableofcontents

\section{Introduction}
Black holes are thought to be able to scramble quantum information incredibly quickly~\cite{hayden2007black,sekino2008fast}. 
%This implies that they can scramble  information very quickly, as Nature is very good at scrambling. 
%Indeed, counterintuitively
More broadly, it is known that quantum dynamics can approximately reproduce truly uniformly random behaviour on timescales poly-logarithmic in the system size, even under extremely strong demands of what ``approximation'' should constitute~\cite{schuster2024random,laracuente2024approximate,cui2025unitary,schuster2025strong,parella2026strong}. This is an intrinsically quantum mechanical effect, with such rapid scrambling of information by classical systems straightforwardly disallowed by lightcone arguments~\cite{schuster2024random}.

Within the context of black hole physics, Hayden and Preskill discovered in their seminal work~\cite{hayden2007black}  a particularly remarkable consequence of the scrambling capacities of random quantum systems: quantum information which falls into an old black hole is in a certain sense reflected back out again almost instantaneously. More concretely, the statement goes as follows. We imagine our old black hole to have already radiated away at least half of its initial mass\footnote{We will take ``old'' to be synonymous with ``more than half evaporated''}, with this radiation taken to be maximally entangled with the remaining black hole~\cite{page1993average}. An observer Bob has intercepted this radiation and stored it coherently. A second observer, Alice, then throws an $m$-qubit state (a ``message''), initially maximally entangled with a reference system $A$, into the black hole. The surprising result is that Bob (in principle) can recover Alice's message with high fidelity after collecting only a few more than $m$ additional Hawking radiation qubits. The intuition is that the strongly scrambling nature of the black hole dynamics quickly decouples the reference system $A$ from the degrees of freedom of the remaining black hole, leaving Alice's message recoverable from the combined system of the previously and newly emitted Hawking radiation (with all of this to be made precise in Section~\ref{sec:prelim}). 

The Hayden-Preskill  thought experiment invites various natural follow-up questions. For example, the guarantee of Bob's ability to retrieve Alice's message is information theoretic, with the actual computation seemingly generically intractable, though not entirely devoid of non-trivial structure~\cite{yoshida2017efficient,leone2022retrieving}. It is  also interesting to investigate if  the conclusions of the thought experiment require the interior of the black hole to be behaving genuinely Haar randomly\footnote{Up to the second moment anyway, as it is immediate from   Ref.~\cite{hayden2007black} that being a unitary 2-design is already sufficient}, or whether similar effects can be observed under more modest assumptions~\cite{nakata2024hayden,rampp2024hayden,nakata2023black}. 

In this work we consider the Hayden-Preskill  thought experiment under the assumption that the interior dynamics of the black hole can be modelled as a Haar-random \textit{matchgate circuit}; that is, random free fermionic evolution~\cite{jozsa2008matchgates}. This may in  some sense seem rather counter to the spirit of the experiment, as it is known that matchgate circuits are not good scramblers of quantum information. Indeed, for example, they generate only diffusive operator   spreading~\cite{dias2021diffusive}, as opposed to the ballistic spreading of the full unitary case~\cite{nahum2018operator}, and lead to slower (if eventually  volume-law) entanglement growth~\cite{bianchi2021page,paviglianiti2026emergence}; moreover, they support design formation only in linear depth~\cite{west2025no,grevink2025will,west2026ambient}, exponentially slower than the aforementioned logarithmic scaling of the unitary case. Nonetheless, we discover that their analysis reveals some interesting  details that are perhaps not entirely transparent in the original unitary setup.

Our first finding is that such free fermionic black holes do not enable the fast recovery of infalling quantum information. This is unsurprising; indeed, as discussed, matchgate circuits are known to fail to scramble efficiently, thereby lacking the very capacity putatively responsible for the fast recovery of information in the unitary case. Nevertheless, the precise (representation theoretic) mechanism by which such black holes fail to quickly release information is interesting: they act as \textit{``polarised mirrors''}, suppressing the release of information stored in each matchgate irrep according to its fermionic-degree, with low-degree sectors being released significantly more slowly than high-degree sectors. Consequently, we prove that high-fidelity information recovery is impossible by collecting only a few additional radiation qubits. 

Our second (and much more surprising) finding is that, by augmenting the dynamics with a fixed Clifford unitary ``at the horizon'' of the black hole, we can not only restore the fast release of information, but also show that such augmented black holes can outperform their Haar-random counterparts, in the sense of requiring asymptotically fewer radiation qubits to achieve the same target average recovery fidelity. This is despite the fact that prepending a fixed unitary to each element of an ensemble has no effect on several common diagnostics of scrambling; for example, it does not change the so-called frame potential~\cite{mele2023introduction}. As we shall see, therefore, our (fixed) Clifford + (random) matchgate ensemble is very far from being a unitary 2-design, and nevertheless reproduces (and in fact in a sense betters) the results of the Hayden-Preskill thought experiment. 

The remainder of the paper is structured as follows. We begin in Section~\ref{subsec:hp_protocol} by reviewing the Hayden--Preskill protocol and establishing our notation. In Section~\ref{subsec:matchgates}, we introduce the relevant properties of matchgate dynamics, including the explicit structure of the two-copy commutant underlying our main results. In Section~\ref{sec:ffbh}, we derive a degree-resolved decoupling formula and show how low-degree message correlations obstruct high-fidelity recovery in the early-radiation regime. In Section~\ref{sec:nongaussianhorizon}, we show that a suitably chosen fixed Clifford unitary restores Haar-like recovery and provide an explicit construction that asymptotically improves on the Haar-random benchmark. We conclude with a discussion of our results in Section~\ref{sec:discussion}. Proofs of all theorems and lemmas presented in this work are collected in Appendix~\ref{sec:proofs}.

\section{Preliminaries}\label{sec:prelim}
In this section, we establish some notation and review the basic details of both the Hayden-Preskill thought experiment, and some relevant details of the theory of matchgate circuits.

\subsection{The Hayden-Preskill protocol}\label{subsec:hp_protocol}

We begin by introducing the  setup of the Hayden-Preskill protocol~\cite{hayden2007black} (illustrated in Figure~\ref{fig:setup}) and recalling the main conclusions. At the beginning of the thought experiment, our total system is partitioned into four subsystems. We let $B_i$ denote the initial ($n_{B_i}$-qubit) black hole, which we take to be maximally entangled with previously emitted (``early radiation'') $E$~\cite{page1993average}. This system $E$ is imagined to be under the control of an observer Bob, who can perform arbitrary quantum computations on it. A second observer, Alice, possesses a maximally entangled state $\Phi^{MA}$ across two further systems, the $m$-qubit ``message system'' $M$, and a reference system $A$, so that the initial state of the total system is\footnote{Up to a possible unimportant unitary acting on $E$, which we could anyway imagine Bob undoing} $\rho^{EB_iMA}=\Phi^{EB_i}\ot\Phi^{MA}$.  
We assume throughout that the message is small compared with the initial black hole, i.e. $m\ll n_{B_i}$, and use $d_S$ to denote the dimension of a subsystem $S$. For example, $d_{B_i}=2^{n_{B_i}}$ and $d_{A}=d_{M}=2^{m}$.

\begin{figure}
    \centering
    \includegraphics[width=0.6\linewidth]{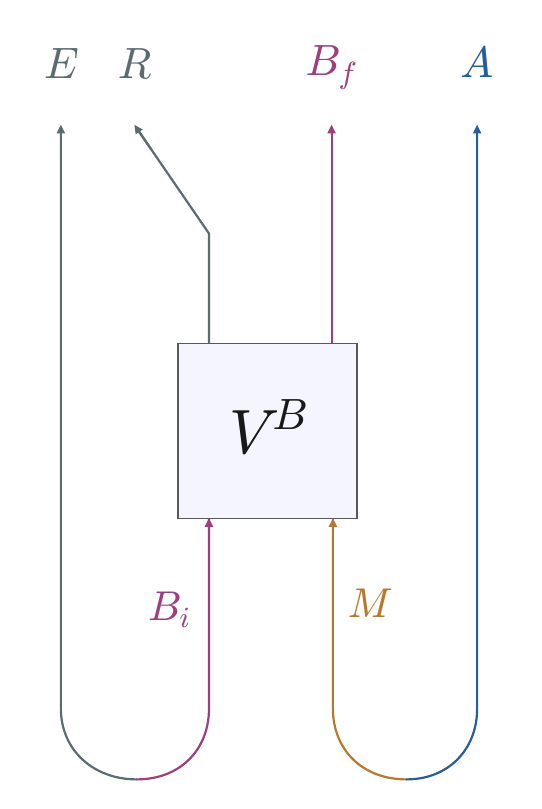}
    \caption{A graphical depiction of the setup of the Hayden--Preskill protocol~\cite{hayden2007black}. The old black hole $B_i$ is assumed to be past the Page time, so that it is modelled as maximally entangled with previously emitted radiation $E$~\cite{page1993average}. Alice possesses a Bell state across the system $MA$, and throws the ``message'' half of it into the black hole. The internal scrambling dynamics of the black hole are modelled by a unitary $V^B$ acting on the combined system $B = MB_i$. Finally, further modes of radiation constituting a system $R$ are emitted, leaving the final black hole $B_f$.  }
    \label{fig:setup}
\end{figure}

Now, we imagine Alice throwing the qubits of $M$ into  the  initial black hole, thereby forming the enlarged black hole $B:= MB_i$, with $n_B=m+n_{B_i}$ and $d_B=d_Md_{B_i}$.  
This black hole then evolves under  dynamics    modelled by a unitary $V^B$ drawn Haar randomly from some compact group $G$  (which in the original Hayden-Preskill protocol was simply taken to be the unitary group $\mbu(d_B)$).  Finally, we imagine the black hole radiating a system $R$ of $r$ further qubits (which are collected by Bob and coherently stored alongside the  previously captured radiation $E$), leaving behind a remaining black hole $B_f$, with $n_B=n_{B_f}+r$ and $d_B=d_{B_f}d_R$. For a given $V^B\in G$, the final state of the total system is therefore
\begin{equation}\label{eq:finalpure}
    \ket{\Psi^{ERB_fA}_{V}}\! =\! (\id_E\ot V^B\ot \id_A)(\ket{\Phi^{EB_i}}\ot \ket{\Phi^{MA}})
\end{equation}
(where it is convenient to identify the unitary $V^B$ as a map  $B\to B_fR$). The corresponding reduced state on $AB_f$ is
\begin{align}
\sigma_V^{AB_f}&:=\Tr_{ER}\!\left[\Psi_V^{ERB_fA}\right] \notag\\
&=\Tr_R\!\left[(\id_A\otimes V^B)\,\omega_{\mathrm{HP}}^{AB}\,(\id_A\otimes(V^B)^\dagger)\right],\label{eq:sigma_abf}
\end{align}
where $\Psi_V^{ERB_fA}:=\ket{\Psi_V^{ERB_fA}}\!\bra{\Psi_V^{ERB_fA}}$, and $\qquad\omega_{\mathrm{HP}}^{AB}:=\Phi^{MA}\otimes\tau^{B_i}$.
We note that $\ket{\Psi^{ERB_fA}_{V}}$ is maximally entangled across the bipartition $ERB_f:A$, so that Alice's final state is the maximally mixed state $\tau^A:=\id^A/d_A$. This is an important observation because, by Uhlmann's theorem, if the reduced final state $\sigma_{V}^{AB_f}$ is (approximately) decoupled across the bipartition $A:B_f$, there exists a decoding isometry acting on $ER$ whose output contains a register
$\widehat M\cong M$ that (after discarding an auxiliary register) is approximately maximally entangled with $A$. 

A key conceptual point (which in some sense could be considered the motivation for taking Alice's message to be maximally entangled with a reference system in the first place) is that perfect recovery in the Bell state case implies that Bob's decoder will automatically recover \textit{any} message state. To see  this, let $\mct_{V,\mathcal D}:M\to\widehat{M}$ denote the entire effective channel (i.e. consisting of the evolution, decoder $\mathcal D$, and partial trace over $JB_f$), so that $(\mct_{V,\mathcal D}\otimes\id_A)(\Phi^{MA})=\Phi^{\widehat{M} A}$. Expanding the Bell-state density operator in the computational basis then gives
\begin{equation*}
    \frac{1}{d_M}\sum_{i,j}\mct_{V,\mathcal D}\big(\ketbra{i}{j}^M\big) \otimes \ketbra{i}{j}^A=\frac{1}{d_M}\sum_{i,j}
\sketbra{ i}{ j}^{\widehat M}
\otimes \ketbra{i}{j}^A\,.
\end{equation*}
But, since the matrix units on $A$ are linearly independent, this implies that $\mct_{V,\mathcal D}(\ketbra{i}{j}^M)=\sketbra{i}{ j}^{\widehat M}$ for every $i,j$.  In particular, if $\ket\psi=\sum_i\alpha_i\ket{i} $,  then 
\begin{equation*}
\mct_{V,\mathcal D}(\ketbra{\psi}{\psi}^M)=\sum_{i,j}\alpha_i\alpha_j^*\sketbra{ i}{ j}=\sketbra{\psi}{\psi}^{\widehat M}.
\end{equation*}

\begin{figure*}
    \centering
    \includegraphics[width=0.9\linewidth]{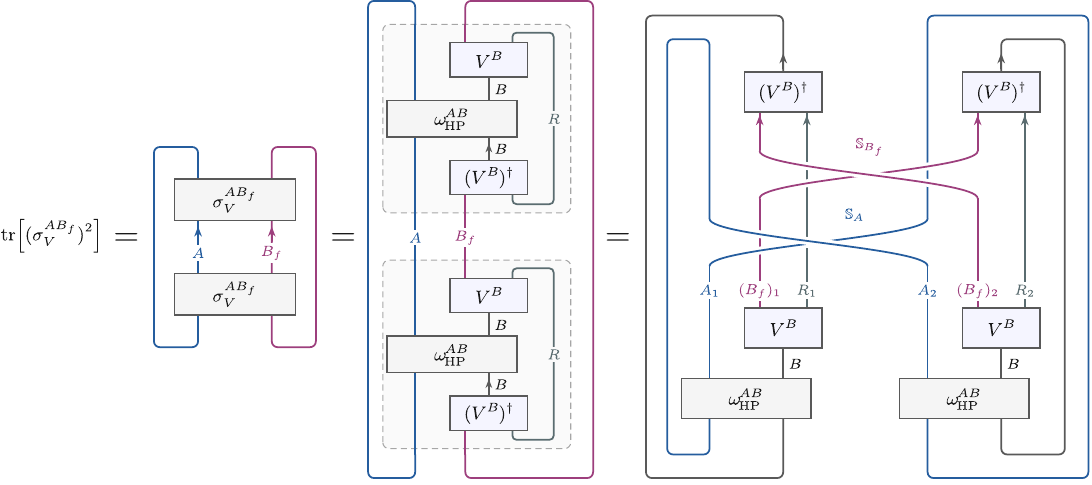}
    \caption{A tensor network diagram depicting the reduction, via the swap trick~\cite{mele2023introduction}, of the purity $\tr\left[(\sigma_{V}^{AB_f})^2\right]$ to a form amenable to a second-order integral over $G$. Here, the reduced state is defined in Eq.~\eqref{eq:sigma_abf}, and the initial state on $MA$ is the Bell state. Reading the final diagram from bottom to top yields Eq.~\eqref{eq:pure_swap}. }
    \label{fig:purity}
\end{figure*}

Thus, in order to quantify Bob's ability to recover Alice's quantum message, we study his ability to restore the maximally entangled state $\Phi^{\widehat MA}$. To that end, let us define the optimal \textit{recovery fidelity}  
 
\begin{align*}
F_{\mathrm{rec}}(V^B)&:=\max_{\mct_{V,\mathcal D}}F\!\left((\mct_{V,\mathcal D}\otimes\id_A)(\Phi^{MA}),\Phi^{\widehat M A}\right)\,,
\end{align*}
where $F(\rho,\sg)=\|\sqrt{\rho}\sqrt{\sg}\|_1^2$ denotes the usual fidelity of quantum states and we optimise over all possible CPTP maps implementable by Bob's actions on $ER$. Equivalently, we can define the corresponding recovery infidelity by $\varepsilon:=1-F_{\mathrm{rec}}(V^B)$.

Our   intuition connecting   the decoupling of $\sigma_{V}^{AB_f}$ over the bipartition $A:B_f$ with Bob's capability to reconstruct Alice's message may  then be formalised in several useful ways. First, we show in Appendix~\ref{app:recovery_general} that we have
\begin{restatable}{lem}{lemsep}\label{lem:sep}
For every $V^B$, the  optimal recovery fidelity may be expressed as 
\begin{equation}\label{eq:fid_opt}
 F_{\mathrm{rec}}(V^B) =\max_{\zeta^{B_f}}
 F\big(\sigma_{V}^{AB_f},\tau^A\otimes\zeta^{B_f}\big)\,,
\end{equation}
where we maximise over all states $\zeta^{B_f}$ on $B_f$. 
\end{restatable}

We also find it useful to introduce an operator $X_{V}^{AB_f}$ capturing the failure of the final state on $AB_f$ to decouple from its marginals, namely
\begin{align}
    X_{V}^{AB_f}&=\sigma_{V}^{AB_f}-\sigma^A\otimes\sigma_{V}^{B_f}\nonumber\\
    &=\sigma_{V}^{AB_f}-\tau^A\otimes\tau^{B_f},\label{eq:xab} 
\end{align}
where the fact that  $\sigma^{B_f}_{V}=\tr_A[\sigma_{V}^{AB_f}]$ is the maximally mixed state $\tau^{B_f}$ on the remaining black hole follows from a  direct calculation. We are therefore led to consider what one might call the 
\textit{average decoupling error} over $V^B\sim\mu_G $, i.e. 
\begin{equation}\label{eq:dg}
    \Delta^{G}_p:=\expect_{V^B\sim \mu_G}\big\|X_{V}^{AB_f}\big\|_p^2\, ,
\end{equation}
where $\|(\cdot)\|_p$ denotes the Schatten $p$-norm. It will sometimes be helpful to use that, by the Fuchs–van de Graaf inequality~\cite{fuchs1999cryptographic}, we have
\begin{equation}\label{eq:fdg}
    \|X_{V}^{AB_f}\|_1 \leq 4\sqrt{1-F_{\rm rec}(V^B)}.
\end{equation}

Now, a second formalisation of the principle linking the decoupling of the state over $AB_f$ with Bob's capability to reconstruct Alice's message is then furnished by Lemma 1.1 of Ref.~\cite{hayden2008decoupling}, which implies that, for every $V^B$, we have 
\begin{align*}
F_{\mathrm{rec}}(V^B)&\geq1-\big\|X_{V}^{AB_f}\big\|_1 \,.
\end{align*}
Consequently, we have by Jensen's inequality,   the concavity of the square root  function, and Eq.~\eqref{eq:dg}, that  
\begin{align}
\expect_{V^B\sim \mu_G} F_{\mathrm{rec}}(V^B)&\geq 1-\expect_{V^B\sim \mu_G} \big\|X_{V}^{AB_f}\big\|_1\nonumber\\
&= 1-\expect_{V^B\sim \mu_G} \sqrt{\big\|X_{V}^{AB_f}\big\|_1^2}\nonumber\\
&\geq 1-\sqrt{\Delta^{G}_1} \label{eq:fidelity_bound}
\end{align}
Thus, $\Delta^{G}_1\leq\varepsilon^2$ guarantees an average recovery fidelity of at least $1-\varepsilon$. For the purposes of upper bounding $\Delta^{G}_1$ (so as to obtain guarantees on the average recovery fidelity), it can be sufficient to employ the standard estimate  $\|X\|_1^2\leq d\|X\|_2^2$, which holds for any operator   $X:\mbc^d\to \mbc^d$. In our present circumstances, this leads to
\begin{align}\label{eq:general_upper_bound}
    \Delta^{G}_1&\leq d_Ad_{B_f}\Delta^{G}_2\\
    &= d_Ad_{B_f}\expect_{V^B\sim \mu_G} \tr[\big(\sigma_{V}^{AB_f}-\tau^A\otimes\tau^{B_f}\big)^2]\nonumber\\
    &= d_Ad_{B_f}\expect_{V^B\sim \mu_G} \tr[\big(\sigma_{V}^{AB_f}\big)^2]-1\,.\nonumber
\end{align}
The utility of this reformulation is that $\expect_{V^B\sim \mu_G} {\rm tr}[(\sigma_{V}^{AB_f})^2]$, describing the average purity of  $\sigma_{V}^{AB_f}$, may be evaluated via \textit{Weingarten calculus}~\cite{mele2023introduction,weingarten1978asymptotic,collins2006integration} over $G$. 
Indeed, using Figure~\ref{fig:purity}, one has

\begin{equation}\label{eq:pure_swap}
\begin{aligned}[b]
&{\rm tr} [(\sigma_V^{AB_f})^2 ]\! =\!{\rm tr}\Bigl[(\omega_{\mathrm{HP}}^{AB})^{\otimes 2}\bigl(\mbs_A\!\otimes\!V^{B \dagger\otimes 2}(\mbs_{B_f}\!\otimes\!\id_R^{\otimes 2})V^{B \otimes 2}\bigr)\Bigr],
\end{aligned}
\end{equation}
where $\mbs_X$ denotes the swap operation between two isomorphic copies of $X$. For example, when  $G=\mbu(d_B)$ (i.e. the setting of the original Hayden-Preskill protocol) one   quickly finds  for fixed $m,r$ and $n_B\to\infty$
\begin{equation}\label{eq:d2u}
 d_Ad_{B_f} \Delta_2^{\mbu(d_B)} = \frac{(d_A^2-1)(d_{B_f}^2-1)}{d_B^2-1}\sim (4^m-1)4^{-r}\,,
\end{equation}
whence Eqs.~\eqref{eq:fidelity_bound} and~\eqref{eq:general_upper_bound} combine to yield $\expect_{V^B\sim \mu_{\mbu(d_B)}} F_{\mathrm{rec}}(V^B)\gtrsim 1-2^{m-r}$, which is the bound of Ref.~\cite{hayden2007black}. Thus, for a fully unitary Haar-random black hole, Bob needs to collect only a few more qubits than the size of the message in order to recover the message with a fidelity exponentially close to one.

Thus, analysis at the level of using the estimate $\Delta_1^{G}\leq d_Ad_{B_f} \Delta_2^{G}$ already suffices to yield the core result in the case of a unitary Haar-random black hole. More generally, however, we would like both to tighten the bound of Eq.~\eqref{eq:fidelity_bound} as much as possible, and prove an \textit{upper} bound on the fidelity, so as to certify the failure of certain black holes to reflect information. To that end, our first main result, proved in Appendix~\ref{app:recovery_general}, is the following.
\begin{restatable}{thm}{thmnew}\label{thm:new}
%For all compact groups $G$, 
Let $\mu$ be any probability measure on the unitaries acting on $B$. The average optimal recovery fidelity satisfies
\begin{equation}
 \frac{1}{1+d_Ad_{B_f} \Delta_2^{\mu}}\leq \expect_{V^B\sim\mu}F_{\mathrm{rec}}(V^B)\leq1-\frac{d_Ad_{B_f} \Delta_2^{\mu}}{4(d_A+1)^2},
 \label{eq:newbound}
\end{equation}
where $\Delta_2^\mu:=\expect_{V^B\sim\mu}\bigl\|X_V^{AB_f}\|_2^2$.
\end{restatable}
The theorem applies in particular when $\mu=\mu_G$ is the Haar measure on a compact group $G$, in which case we write $\Delta_p^G:=\Delta_p^{\mu_G}$. We state the result in this slightly more general form because we will later apply it to our Clifford-augmented ensemble. 

Theorem~\ref{thm:new} allows us to improve the bound on the message recovery fidelity in the  $\mbu(d_B)$ case. From Eq.~\eqref{eq:d2u} and the identity $d_Ad_{B_i}=d_Rd_{B_f}$ we see immediately that
\begin{align*}
    \expect_{V^B\sim\mu_{\mbu(d_B)}}F_{\mathrm{rec}}(V^B) \gtrsim \frac{1}{1+d_A^2d_{B_f}^2/d_B^2}\gtrsim 1 - 4^{m-r}\,,
\end{align*}
representing a quadratic improvement on the infidelity bound of Ref.~\cite{hayden2007black}. 

Finally, in parts of our analysis it will be useful to allow the
initial state on $AB$ to be some arbitrary state  $\omega^{AB}$; indeed, this will reveal a more structural perspective on the black-hole dynamics in our later analysis. In these cases, we include an explicit $\omega$-label to the notation of our previously defined quantities in order to indicate the analogous property in the case of an initial state $\omega^{AB}$; for example, Eq.~\eqref{eq:xab} becomes 
\begin{equation}\label{eq:x_arbitrary}
X_{V,\omega}^{AB_f}=\sigma_{V,\omega}^{AB_f}-\omega^A\otimes\sigma_{V,\omega}^{B_f},
\end{equation}
and Eq.~\eqref{eq:dg} becomes
\begin{equation}\label{eq:dgw}
\Delta^{G,\omega}_p:=\expect_{V^B\sim \mu_G}\big\|X_{V,\omega}^{AB_f}\big\|_p^2\,.
\end{equation}
However, unless an explicit input label is displayed, we always
refer to the Hayden--Preskill input $\omega_{\mathrm{HP}}^{AB}=\Phi^{MA}\otimes\tau^{B_i}$. In particular,
all recovery statements in this work refer to the Hayden--Preskill input $\omega_{\mathrm{HP}}^{AB}$.

\subsection{Matchgate circuits}\label{subsec:matchgates}

In this section, we review the necessary representation theory of the matchgate group $\mathsf{MG}(n)\cong{\rm Spin}(2n)$, and the associated properties of its commutant. Further details can be found in Refs.~\cite{braccia2026commutant,sierant2026theory}. Throughout this subsection, we will denote our $n$-qubit Hilbert space by $\mch$ with dimension $d:=2^n$. 
We will  use the standard choice,
\begin{align}\label{eq:majos}
\gamma_1&=XIII\cdots I\;,&\gamma_{2}&=YIII\cdots I\nonumber\\
\gamma_3&=ZXII\cdots I\;,&\gamma_{4}&=ZYII\cdots I\\
&\:\:\vdots&&\:\:\vdots\nonumber\\
\gamma_{2n-1}&=ZZ\cdots ZX\;,&\gamma_{2n}&=ZZ\cdots ZY\;,  \nonumber
\end{align}
of Jordan-Wigner Majorana operators, which obey the canonical anticommutation relations $\{\gamma_\mu,\gamma_\nu\}=2\delta_{\mu\nu}\id$. For a subset $S=\{\mu_1<\cdots<\mu_k\}\subseteq [2n]$, we define the Majorana string $\gamma_S=i^{k(k-1)/2}\gamma_{\mu_1}\gamma_{\mu_2}\cdots\gamma_{\mu_k}$, where the factors of $i$ ensure that $\gamma_S$ is Hermitian.  For each $k\in\{0,\ldots,2n\}$, we denote the set of degree-$k$ Majorana monomials by $\mcb_k^{(n)}:=\{\gamma_S:S\subseteq[2n],\ |S|=k\}$, so that $|\mcb_k^{(n)}|=\binom{2n}{k}$; up to signs, these monomials constitute exactly the set of Pauli strings, and thus form a (Hermitian) orthogonal basis of the operator space.

A matchgate unitary $U_O\in\mathrm{Spin}(2n)$\footnote{The subtlety  that the matchgates act as a faithful representation of $SO(2n)$ on the Majoranas but only as a faithful representation of ${\rm Spin}(2n)$ (and a projective representation of $SO(2n)$) on $\mch$ does not cause any trouble in this work.} acts linearly on the Majoranas,
\begin{equation}
U_O\gamma_\mu U_O^\dagger=\sum_{\nu=1}^{2n}O_{\mu\nu}\gamma_\nu \,,\label{eq:majo_linearly}
\end{equation}
for some corresponding $O\in\mbso(2n)$. By multilinearity and the canonical anticommutation relations, this induces the standard exterior-power action
\begin{equation}
U_O\gamma_S U_O^\dagger=\sum_{\substack{T\subseteq[2n]\\ |T|=k}}\det(O_{S,T})\,\gamma_T=\sum_{\substack{T\subseteq[2n]\\ |T|=k}}(\Lambda^k O)_{S,T}\,\gamma_T,\label{eq:ext}
\end{equation}
where $|S|=k$ and $O_{S,T}$ denotes the submatrix of $O$ with rows indexed by $S$ and columns indexed by $T$. Thus, the subspace $V_k:=\operatorname{span}\mathcal B_k^{(n)}$ is invariant under matchgate conjugation and carries the exterior-power representation $\Lambda^k(\mathbb C^{2n})$. Consequently,
\begin{equation}
    \operatorname{End}(\mch)\cong\bigoplus_{k=0}^{2n}\Lambda^k(\mathbb C^{2n})
\end{equation}
as a $\mathrm{Spin}(2n)$ representation under conjugation.

Now, we use two representation theoretic facts. The first is that the $V_k:=\Lm^k(\mathbb C^{2n})$ are, for $k\neq n$, irreducible and pairwise non-isomorphic except to their respective \textit{Hodge duals}, $\Lm^k(\mathbb C^{2n})\cong \Lm^{2n-k}(\mathbb C^{2n})$; when $k=n$ one has the decomposition $V_n\cong V_{n,+}\oplus V_{n,-}$ into two non-isomorphic irreps corresponding to the kernels of $(\id\pm Z\tn)\vert_{V_n}$~\cite{fulton1991representation}.
The second fact that we need is that, by Schur's lemma~\cite{fulton1991representation}, the second order commutant of the matchgate group is fixed by the previous discussion to be
\begin{equation}
    \mcc^{(2)}_{\mg} \cong \bigoplus_{k=0}^{n-1} \left[\id_{V_k} \ot {\rm End}\, \mbc^2\right] \oplus \mbc\id_{V_{n,+}} \oplus \mbc\id_{V_{n,-}}\, ,\label{eq:comm}
\end{equation} 
and is evidently thereby of dimension $4n+2$. 

Let $\Pi_n:=Z^{\otimes n}$ be the fermionic parity operator. An orthonormal basis of this space that will turn out to be particularly convenient for us is given by~\cite{diaz2023showcasing,west2025graph}
\begin{align}
    Q_{k,0}&:=\frac{1}{d\sqrt{\binom{2n}{k}}}\sum_{P\in \mathcal{B}_k^{(n)}}P\otimes P,\label{eq:comm_basis_0}\\
    Q_{k,1}&:=\frac{i^{k\bmod 2}}{d\sqrt{\binom{2n}{k}}}\sum_{P\in \mathcal{B}_k^{(n)}}P\otimes\Pi_n P,\label{eq:comm_basis_1}
\end{align}
for $k=0,\ldots,2n$, where $d=2^n$ and the phase in $Q_{k,1}$ ensures that these operators are Hermitian also for odd $k$.
We will also make use of the so-called $\ell$\textsuperscript{th}-order \textit{frame potential} $F_{G}^{(\ell)}$, which for  a compact group $G$ is given by~\cite{roberts2017chaos, mele2023introduction}
\begin{equation}\label{eq:fp}
F_{G}^{(\ell)} = \expect_{U,V\sim \mu_G} \left[ \left\lvert \tr\left(UV^\dg\right)\right\rvert^{2\ell}\right].
\end{equation}
$F_{G}^{(\ell)}$  measures the 2-norm  distance between the $\ell$-fold moment operators of $\mu_G$ and that of a unitary $\ell$-design~\cite{mele2023introduction}; one can also show~\cite{west2025graph} that, 
given a decomposition 
\begin{equation}
   \mch^{\otimes \ell} = \bigoplus_{\lm} \bigoplus_{i=1}^{m_{\lm}} \mch_{\lm,i}\;,
\end{equation}
into $G$-irreps, the corresponding  $\ell$\textsuperscript{th}-order frame potential satisfies
\begin{equation}\label{eq:fpr}
F_{G}^{(\ell)} = \sum_{\lm}  m_{\lm}^2,
\end{equation}
with $ m_{\lm}$ the multiplicity of the irrep $\lm$ in $\mch^{\otimes \ell}$.  Combining these observations yields that $F^{(2)}_{\mg} = 4n+2$. As it follows immediately from Eq.~\eqref{eq:fpr} that $F^{(2)}_{\mbu(d)} = 2$, we see that for $n\gg 1$ the matchgate group is far from being a unitary 2-design.

\section{Free fermionic  black holes}\label{sec:ffbh}

Suppose that the nearest large black hole into which Alice can toss her message happens to be a  free fermionic black hole. That is, the black hole dynamics are modelled by a unitary $V^B$ drawn Haar randomly from the matchgate group $\mgb$. We fix the input and output Jordan--Wigner orderings such that $M$ and $B_f$ occupy the first $m$ input modes and the first $n_{B_f}$ output modes, respectively.
As discussed in the preliminaries, Bob's ability to recover Alice's message will primarily depend on the free fermionic black hole's ability to decouple the reference system $A$ from the remaining black hole $B_f$. To gain structural insight into its decoupling behaviour, we begin by evaluating the averaged Hilbert--Schmidt correlation $\Delta_2^{\mgb,\omega}$ (recall Eqs.~\eqref{eq:x_arbitrary} and~\eqref{eq:dgw}) for an arbitrary input state $\omega^{AB}$. 
In this pursuit, it is useful to define a few quantities. First, we introduce
\begin{align}
    p_k(n_{B_f},n_B)&=\frac{\binom{2n_{B_f}}{k}}{\binom{2n_B}{k}}\sim (n_{B_f}/n_B)^k,\label{eq:p}
\end{align}
where the final asymptotic relation holds for a fixed $k$ as $n_B,n_{B_f}\to\infty$ with $n_{B_f}/n_B$ fixed.
$p_k(n_{B_f},n_B)$ is exactly the fraction of degree-$k$ Majorana monomials whose  support is contained within $B_f$, and is thus the fraction of the (2-norm) weight of a degree-$k$ Majorana monomial that finds itself supported on $B_f$ after the action of a random  matchgate.   Secondly, we introduce
\begin{align}\label{eq:wdef}
    \mcw_k(\omega)&=\sum_{P\in \mcb_k^{(n_B)}}\left\|\operatorname{Tr}_B\left[D^{AB}(\id_A\otimes P)\right]\right\|_2^2,
\end{align}
where $D^{AB}:=\omega^{AB}-\omega^A\otimes\omega^B$ is the difference of the initial state from its marginals. Thus, $\mcw_k(\omega)$ quantifies the initial $A:B$ correlation weight carried by the degree-$k$ Majorana sector.

With these ingredients in hand, we find
\begin{restatable}{thm}{thmmain}\label{thm:main} 
    For an arbitrary input state $\omega^{AB}$, the decoupling error $\Delta_2^{\mgb,\omega}$ admits the  decomposition
    \begin{equation}\label{eq:thm2}
    \Delta_2^{\mgb,\omega}=\frac{1}{d_{B_f}}\sum_{k=1}^{2n_{B_f}}p_k(n_{B_f},n_B)\,\mcw_k(\omega). 
    \end{equation}
\end{restatable}
The proof of Theorem~\ref{thm:main} (see Appendix~\ref{app:polarised_mirror}) is a relatively simple exercise in the Weingarten calculus on the matchgate group; one essentially employs the rearrangement of Figure~\ref{fig:purity} (albeit with minor complications introduced by the fact that $\omega^{AB}$ need not be the maximally entangled input state), and evaluates the resulting second-order integral using the matchgate commutant basis introduced in Eqs.~\eqref{eq:comm_basis_0} and~\eqref{eq:comm_basis_1}.

Let us try to interpret Theorem~\ref{thm:main}, recalling that a small value of $d_Ad_{B_f}\Delta_2^{\mgb,\omega}$ corresponds to successful decoupling. When is the right hand side of Eq.~\eqref{eq:thm2} small? First, from Eq.~\eqref{eq:wdef}, $\mcw_k(\omega)$  can itself be zero if the initial state is already fully decoupled, which leads (not unreasonably) to the black hole successfully ``fully decoupling'' the already-decoupled state. 
Less trivially, since $n_{B_f}<n_B$ and recalling Eq.~\eqref{eq:p}, the factor $({n_{B_f}}/{n_B})^k$ suppresses any correlations across $A:B$ present in the initial state, with the extent of the suppression dictated by the degree $k$. 
%this ratio is strictly smaller than one. Consequently, the survival factor $p_k(n_{B_f},n_B)$ decreases monotonically with the Majorana degree-$k$. 
It is in this sense that we say that the free fermionic black holes   act as  \emph{polarised} mirrors. Correlations in high-degree sectors are quickly suppressed, posing no difficulty to recovery, but correlations in low-degree Majorana sectors, as quantified by $\mcw_k$, are preferentially retained in the remaining black hole.

%Back to maximally entangled
To connect the preceding decoupling result to Bob's ability to recover \emph{any} message from free fermionic black holes, we now specialise, as discussed in Section~\ref{sec:prelim}, to Alice's  input state being maximally entangled; i.e., we take $\omega_{\mathrm{HP}}^{AB}=\Phi^{MA}\otimes\tau^{B_i}$. So, from Eq.~\eqref{eq:thm2}, it remains to evaluate the corresponding degree-resolved correlation weights; 
we find 
\begin{restatable}{lem}{lemw}\label{lem:w} 
The  state $\omega^{AB}_{\mathrm{HP}}$ satisfies $\mcw_k(\omega^{AB}_{\mathrm{HP}})=\frac{1}{d_M}\binom{2m}{k}$ for $1\le k\leq 2m$, while $\mcw_k(\omega^{AB}_{\mathrm{HP}})=0$ for $k>2m$.
\end{restatable}
Crucially, the low-degree correlations of $\omega^{AB}_{\rm HP}$ therefore do not vanish, and will act as obstructions to fast information recovery from a  free fermionic black hole. 
This allows us to construct an explicit lower bound on $\Delta_1^{\mgb}$. Recalling the duality between the trace norm and the operator norm for every Hermitian operator $Y_{V}$ satisfying $\|Y_{V}\|_{\infty}\le 1$, we have $\|X_{V}^{AB_f}\|_1\ge\lvert\tr\big[Y_{V}X_{V}^{AB_f}\big]\rvert$. Since  $\tr\,(Y_{V}X_{V}^{AB_f}\,)=\tr\,(Y_{V}\sigma_{V}^{AB_f})-\tr\,(Y_{V}(\tau^A\otimes\tau^{B_f}))$, we see that $Y_{V}$ acts as  a ``correlation witness'': its overlap with $X_{V}^{AB_f}$ measures the difference between its expectation value in the actual output state and in the corresponding product state. 

So, a good lower bound on $\Delta_1^{\mgb}$ may be obtained by finding an appropriate witness observable $Y_{V}$. Since we learned from Theorem~\ref{thm:main} that degree-one correlations will be suppressed slowest, 
%and the maximally entangled input state $\Phi^{MA}$ contains perfectly normalised correlations between each message Majorana and its partner on $A$, 
we use a witness that detects a degree-one Majorana correlation. With the details given in Appendix~\ref{app:no_recovery}, this leads to:

\begin{restatable}{lem}{lemhp}\label{lem:mg_lower_bound} 
The mean-squared decoupling error over $V^B\sim\mu_{\mgb}$ satisfies $\Delta_1^{\mgb}\geq1-{r}/{n_B}$.
\end{restatable}

Combining Lemma~\ref{lem:mg_lower_bound} with Eq.~\eqref{eq:fdg} already confirms the inability of free fermionic black holes to reflect information for $r\ll n_B$. However, using Lemma~\ref{lem:mg_lower_bound}, we can considerably strengthen the upper bound on the recovery fidelity. For convenience, define
\begin{equation}\label{eq:r}
R_{\mathrm{MG}}:=d_Ad_{B_f}\Delta_2^{\mgb}=\sum_{k=1}^{2m}p_k(n_{B_f},n_B)\binom{2m}{k}.
\end{equation}

\noindent
We then finally find
\begin{restatable}{thm}{thmmgfidelity}\label{thm:main_3} 
For a free fermionic black hole, the average optimal recovery fidelity for Alice's maximally entangled input state $\Phi^{MA}$ satisfies
\begin{align}
    \frac{1}{1+R_{\mathrm{MG}}}\le\expect_{V^B\sim\mu_{\mgb}}F_{\mathrm{rec}}(V^B)\leq \frac{1}{2}\left(1+\sqrt{\frac{r}{n_B}}\right).
\end{align}
\end{restatable}

In particular, in the early-radiation regime $r=o(n_B)$, even for a computationally unrestricted Bob with access to all early and newly emitted radiation, the optimal fidelity satisfies 
\begin{align}
    \expect_{V^B\sim\mu_{\mgb}}F_{\mathrm{rec}}(V^B)\leq\frac{1}{2}+o(1).
\end{align}
Thus, free fermionic black holes fail to act as Hayden--Preskill mirrors because their degree-preserving dynamics obstruct the suppression of low-degree Majorana correlations. Consistent with this picture, the lower bound of Theorem~\ref{thm:main_3} guarantees recovery with fidelity approaching one whenever $R_{\mathrm{MG}}\rightarrow 0$. Achieving this limit, however, requires Bob to wait until the black hole has almost completely evaporated and to collect essentially all of the emitted radiation.

Finally, let us recall that the Haar-random decoupling factor is suppressed by the number of emitted qubits as $4^{-r}$. In contrast we have here that in the early radiation regime, the suppression factor of the degree-$k$ Majorana sector can be written as 
\begin{equation*}
p_k(n_{B_f},n_B)=\frac{\binom{2n_{B}-k}{2r}}{\binom{2n_B}{2r}}\approx \left(1-\frac{k}{2n_B}\right)^{2r}.
\end{equation*}
Requiring that this reproduces the $2^{-2r}=4^{-r}$ suppression implies
\begin{align}\label{eq:reproduce_haar}
    \left(1-\frac{k_\star}{2n_B}\right)^{2r}\sim 2^{-2r},
\end{align}
which forces $k_\star\approx n_B$. Thus, a Haar-like decoupling requires message correlation to be concentrated in Majorana sectors of extensive degree $k=\Omega(n_B)$. In this case Theorem~\ref{thm:main} combined with Theorem~\ref{thm:new} would guarantee a Haar-like lower bound on the recovery fidelity. 
However, due to the initial black hole being maximally entangled with the early radiation $E$, we have for any initial state $\omega^{AB}=\rho^{AM}\otimes\tau^{B_i}$ that $\omega^{AB}\rvert_{B_i}\propto\id_{B_i}$, with the Pauli expansion of $\omega^{AB}$ therefore being non-trivial only on the $m$ message qubits. This implies that $\omega^{AB}$ has support on sectors of degree at most $2m$; for $m\ll n_B$, we therefore never achieve a Haar-like decoupling, even when distributing correlations differently from the Bell state.

\section{A non-Gaussian horizon}\label{sec:nongaussianhorizon}
As discussed in the previous section, the fundamental obstruction to matchgate black holes efficiently reflecting information is the fact that the input state to the black hole is effectively supported only on irreps of Majorana degree at most $2m\ll n_B$. Let us then consider a somewhat more general model, wherein the interior dynamics consist of the application of a (fixed) Clifford unitary $C$, followed by a Haar-random matchgate. In the following, we denote by $\mu_{\mgb C}$ the pushforward of $\mu_{\mgb}$ under the map $V^B\mapsto V^BC$. Thus, drawing $U^B\sim\mu_{\mgb C}$ means drawing $V^B\sim\mu_{\mgb}$ and setting $U^B=V^BC$. In light of Eq.~\eqref{eq:reproduce_haar}, in order to maximise the reflective ability of the black hole, we therefore seek a Clifford $C$ that maps the message Paulis into sectors of extensive Majorana degree; let us verify that this is possible\footnote{That there is something at all to verify is a result of the fact that the (adjoint) action of the Cliffords on the Paulis (and therefore our Majoranas) is not an arbitrary permutation, but rather a map ${\rm Ad}_C\in \mbsp(2n_B,2)$ that preserves their symplectic structure.}.

In fact, it turns out that most Cliffords have the desired properties detailed in the previous paragraph. To see this, let us take a random one $C$. Up to an irrelevant sign, the image $CPC\ad$ of its action on any fixed (non-identity) Pauli $P$ is uniformly distributed on the set of  $4^{n_B}-1$ (non-identity) Majorana monomials, and can be identified with a (nonempty, but otherwise uniformly random) subset $S\subseteq [2n_B]$, corresponding to the Majoranas which appear in its monomial. Therefore, the degree of the image is essentially\footnote{Again, modulo the impossibility of it being zero; let us henceforth stop mentioning this qualifier} a binomial random variable $K\sim\operatorname{Binomial}(2n_B,1/2)$. Equivalently, with  $X_j\sim \operatorname{Bernoulli}(1/2)$, the random variable describing whether or not the $j$\textsuperscript{th} Majorana appears in the monomial, we have $|S|=\sum_{j=1}^{2n_B}X_j$. For $t>0$, it then follows from Hoeffding's inequality that $\Pr(||S|-n_B|>t)\leq  2e^{-t^2/n_B}$.
Calling a Clifford ``bad'' if it fails to map all of the message Paulis into strings with degrees in the interval $[n_B-t,n_B+t]$, then, we have by the union bound that the probability of a random Clifford being bad is at most $\sim  4^m\exp\left(-{t^2}/{n_B}\right)$, which for constant $m$ is negligible for $t\gg \sqrt{n_Bm}\sim\sqrt{n_B}$. Choosing $t=n_B^{2/3}$, the probability of a Clifford being bad tends to zero as $n_B\to\infty$ for fixed $m$. Hence, for sufficiently large $n_B$, we can certainly find a Clifford which maps all of the message Paulis into strings with degrees in the interval $[n_B-n_B^{2/3},n_B+n_B^{2/3}]$. This would yield by the lower bound of Theorem~\ref{thm:new} and Eq.~\eqref{eq:reproduce_haar} a Haar-like recovery fidelity in the early-radiation regime.

We emphasise that we are  \textit{not} arguing that applying a random Clifford before the matchgate evolution reproduces the unitary case (which is trivially true as a result of the Cliffords forming a unitary 3-design) but rather using a probabilistic argument to show that there exists a fixed (black-hole-defining, if you like) Clifford such that prepending it to the random matchgate evolution reproduces the behaviour of the unitary case (and in fact, this probabilistic argument turns out to show that most Cliffords suffice for this purpose). 
Indeed, to see the nontrivial content of this result one can, for example, note that premultiplying the uniform ensemble over the matchgates by a fixed Clifford $C$ (or indeed any fixed unitary) does not change its frame potential;   from  Eq.~\eqref{eq:fp} we have
\begin{align*}
F_{GC}^{(\ell)} &= \expect_{U,V\sim \mu_G} \Big[ \left\lvert \tr\left((UC)(VC)^\dg\right)\right\rvert^{2\ell}\Big]\\
 &= \expect_{U,V\sim \mu_G} \Big[ \left\lvert \tr\left(U(CC\ad) V^\dg\right)\right\rvert^{2\ell}\Big]\\
 &= F_{G}^{(\ell)}.
\end{align*}
In particular, already at $\ell=2$ we have for large black holes that $F_{\mgb C}^{(2)}=F_{\mgb }^{(2)}=4n_B+2\gg 2 = F_{\mbu(d_B)}^{(2)}$, so that, at least insofar as low values of the  frame potential are agreed to constitute a meaningful diagnostic of scrambling ability~\cite{roberts2017chaos}, our results suggest that strong scrambling is not required to reproduce the main conclusion of the Hayden-Preskill thought experiment. 

In fact, the preceding conclusion can be strengthened further. Let us now construct an explicit Clifford which, when followed by a random matchgate, requires fewer radiation qubits $r$ to achieve a target recovery fidelity than in the $\mbu(d_B)$ case. To that end, it is useful to  represent a Majorana monomial $\gamma_S$ by a vector  $S\in\mbf_2^{2n_B}$, where a $1$ at position $j$ indicates the presence of the $j$\textsuperscript{th} Majorana in the factorisation of $\gamma_S$ (with respect, say, to the ordering of Eq.~\eqref{eq:majos}). Under this encoding, it is not difficult to see that the commutation relations of two Majorana monomials may be represented as~\cite{bravyi2010majorana} $\gamma_S\gamma_T=(-1)^{b_{2n_B}(S,T)}\gamma_T\gamma_S$, where
\begin{align}\label{eq:sympf}
 b_{2n_B}(s,t)=|s|\,|t|+s\cdot t\pmod2\,.
\end{align}
One can think of this as simply a convenient change of basis from the usual representation of an $n_B$-qubit Pauli string as a vector $(u\,|\,v)\in\mbf_2^{2n_B}$ into $X$ and $Z$ incidences, wherein commuting a string $(u\,|\,v)$ past a second string $(u'\,|\,v')$ incurs a phase $(-1)^{uv'+u'v}$ coming from  the canonical symplectic form. Now (as usual) the Clifford group (modulo the Paulis and phases, anyway) is simply given by the symplectic group which preserves the antisymmetric\footnote{Recall that the distinction between symmetric and antisymmetric vanishes over $\mbf_2$} form $b_{2n_B}$, so that what we are looking for is a symplectic injection $T: (\mbf_2^{2m},b_{2m})\hookrightarrow(\mbf_2^{2n_B},b_{2n_B})$ from the subspace corresponding to the Majorana monomials that act trivially outside of the message qubits, into the full $n_B$-qubit black hole\footnote{\textit{a priori} one might be concerned that an arbitrary symplectic injection might not   extend to a symplectic map on the entire space; happily this is however guaranteed by Witt's theorem~\cite{wood1993witts}}. In particular, we are looking for such a map with the additional property that its image consists only of vectors of Hamming weight roughly $n_B$, so as to be rapidly decoupled from $B_f$ by the subsequent random matchgate.

\noindent
To that end, let us consider the  map 
\begin{equation*}
    {\rm ev}:\mbf_2^{2m}\longrightarrow\mbf_2^{d_M^{2}-1}; \qquad {\rm ev}(x)=(a\cdot x)_{a\in\mbf_2^{2m}\setminus\{0\}}\,.
\end{equation*}
That is, $\rm{ev}$ evaluates its input against each of the  nonzero linear functionals on $\mbf_2^{2m}$. Evidently, $\lvert {\rm ev}(x)\rvert = d_M^2/2 $ for all nonzero $x$. Thus, every nonzero input is mapped to a string with $d_M^2/2$ ones among $d_M^2-1$ positions, strictly more than half its length. We are still a few steps away from having the map we want; for a start, we wish to map into $\mbf_2^{2n_B}$, not $\mbf_2^{d_M^{2}-1}$, and for another we need our map to be Clifford.  In fact, we can solve both of these problems by concatenating $\rm{ev}$ an even number of times, appending the input, and then padding with zeros as required. That is, we consider the (family of) maps $C_t:\mbf_2^{2m}\to\mbf_2^{2n_B}$ which act as 
 \begin{equation}
 C_t(x) = ({\rm ev}(x)^{\oplus 2t},\ x,\ 0^{2n_B-(2t(d_M^2-1)+2m)})\,.\label{eq:rep_ev}
\end{equation}
That $C_t$ is symplectic is ensured by construction: indeed, for $x,y\in \mbf_2^{2m}$ we have
\begin{equation*}
\lvert C_t(x)\rvert  = 2t|{\rm ev}(x)|+|x|+0 \equiv|x|  \pmod2\,
\end{equation*}
and similarly $C_t(x)\cdot C_t(y)\equiv x\cdot y \pmod2$, so that by Eq.~\eqref{eq:sympf} we have
\begin{equation*}
    b_{2n_B}(C_t(x), C_t(y))= b_{2m}(x, y),
\end{equation*}
for all $x,y$. By Witt's extension theorem~\cite{wood1993witts}, we can lift this to a symplectic map $\widetilde{C}_t:\mbf_2^{2n_B}\to \mbf_2^{2n_B}$, and therefore to a corresponding $n_B$-qubit Clifford. Finally, let us take $t$ as large as possible, i.e. $t^*=\lfloor(n_B-m)/(4^m-1)\rfloor$ (see Eq.~\eqref{eq:rep_ev}).  

For every nonzero message label $x$, the resulting image has Hamming weight $|C_{t^*}(x)|=4^m t^*+|x|$. Consequently, for fixed $m$, $|C_{t^*}(x)|/n_B$ approaches $4^m/(4^m-1)>1$ as $n_B\to\infty$. Thus, for sufficiently large $n_B$, the corresponding Clifford unitary maps every nontrivial message Pauli to a Majorana monomial of degree strictly greater than $n_B$.

Applying the machinery of Section~\ref{sec:ffbh} and defining 
\begin{equation}
    q_m:=\frac{2(4^m-1)}{4^m-2}\label{eq:def_qm}
\end{equation}
then yields
\begin{restatable}{thm}{thmaugrec}~\label{thm:augrec}
For fixed $m,r$, and $n_B\to \infty$, we have 
\begin{equation}\label{eq:qm}
    d_Ad_{B_f}\Delta_2^{\mgb{C}_{t^*}}\sim (4^m-1)q_m^{-2r}\,.
\end{equation}
\end{restatable}

Remarkably, $q_m>2$, so the factor $q_m^{-2r}$ appearing in Theorem~\ref{thm:augrec} decays faster with $r$ than the corresponding Haar-random factor $2^{-2r}$. This shows that the corresponding recovery infidelity can be asymptotically smaller than in the Haar-random case (although notice that this will only be meaningfully true for small $m$, with $q_m\to2$ exponentially quickly in $m$). Indeed, from the (upper and lower) bounds of Theorem~\ref{thm:new}, we have average infidelities given by (for a fixed message size $m$ and first taking $n_B\to\infty$ at each fixed $r$, before
letting $r\to\infty$ so that Theorem~\ref{thm:augrec} applies)
\begin{align*}
    \varepsilon_{\mbu(d_B)} &= \Theta(2^{-2r}),\\
    \varepsilon_{\mgb {C}_{t^*}} &= \Theta(q_m^{-2r}).
\end{align*}
Suppose, then, that Bob wishes to recover Alice's state to some target infidelity $\varepsilon$. Solving for $r$ in the same asymptotic regime, the required number of radiation qubits is given by
\begin{align*}
    r_{\mbu(d_B)} &= \frac{\log(1/\varepsilon)}{\log 4} +\mco(1),\\
    r_{\mgb {C}_{t^*}} &= \frac{\log(1/\varepsilon)}{\log q_m^2} +\mco(1).
\end{align*}
In the high fidelity limit $\varepsilon\to 0$, we therefore find a multiplicative advantage for reaching a given target fidelity. For example, consider $m=1$, the case of a single message qubit. One quickly sees from Eq.~\eqref{eq:def_qm} that $q_1=3$, so that the relevant ratio is 
$r_{\mgb{C}_{t^*}}/r_{\mbu(d_B)}\rightarrow(\log 4)/(\log 9)\approx0.63$. Thus, in the high fidelity limit, the Clifford augmented free fermionic setting requires asymptotically only about $63\%$ as many radiation qubits as the fully Haar-random setting. This corresponds to a reduction of approximately $37\%$.

\section{Discussion}\label{sec:discussion}
The speed with which local quantum dynamics can scramble information is quite remarkable, and is in certain senses exponentially faster than that achievable by classical dynamics~\cite{schuster2024random,laracuente2024approximate}. Hayden and Preskill showed that these scrambling capabilities enable old black holes, modelled by Haar-random dynamics, to reflect infalling quantum information almost immediately in their radiation~\cite{hayden2007black}.

%summary 
In this work, we show that free fermionic black holes, modelled by Haar-random matchgate dynamics, act as degree-resolved mirrors, releasing input correlations stored in low-degree matchgate sectors much more slowly than those stored in high-degree sectors. Consequently, we prove that asymptotically perfect average recovery of infalling quantum information requires almost complete evaporation of the black hole. Despite this, we demonstrate that a single suitably chosen fixed Clifford unitary applied before the matchgate dynamics restores Haar-like recovery. This augmentation leaves the ensemble's second frame potential unchanged and far above its Haar value, demonstrating that Haar-like recovery does not require the dynamical ensemble to form a unitary 2-design. Moreover, we give an explicit Clifford construction that achieves the same target average recovery fidelity as Haar-random dynamics with asymptotically fewer radiation qubits. 

Therefore, our results show that the reflective capacities of black holes cannot simply be reduced to a question of their capacity to scramble, at least when the latter is assessed through the frame potential. Our construction also highlights the importance of how the message is encoded relative to the subsequent dynamics; indeed, our fixed Clifford preprocessing overcomes the low-degree obstruction while leaving the subsequent Gaussian dynamics unchanged. The effectiveness of this scheme is consistent with recent work showing that extensive degree Majorana strings develop volume-law operator entanglement under random matchgate dynamics~\cite{white2026exactoperatorcomplexitymeasures}. It would be interesting to find more dramatic realisations of dynamics that show this reflective behaviour despite exhibiting weak scrambling signatures under conventional diagnostics.

Lastly, our recovery guarantees are information-theoretic and do not by themselves provide an efficient decoding procedure. Building on explicit Hayden--Preskill decoding protocols~\cite{yoshida2017efficient}, it would be interesting to investigate whether the Gaussian structure of the dynamics and the freedom in choosing the Clifford encoding allow lower gate counts or circuit depths than a standard implementation of these protocols.

\begin{acknowledgments}
MK acknowledges support from the Melbourne Research Scholarship.
MW acknowledges support by the Laboratory Directed Research and Development (LDRD) program of Los Alamos National Laboratory (LANL) under project number 20260043DR, and by LANL’s ASC Beyond Moore’s Law project. This work was also supported by the Quantum Science Center (QSC), a National Quantum Information Science Research Center of the U.S. Department of Energy (DOE). 
The main ideas of this work are human-generated; we do, however, acknowledge useful technical conversations with  GPT-5.6 and GPT-6.
\end{acknowledgments}

\bibliography{quantum}
\appendix
\onecolumngrid

\section{Proofs} \label{sec:proofs}

\noindent
In this appendix, we give proofs of the results stated in the main text.

\subsection{Recovery fidelity and decoupling}
\label{app:recovery_general}

We begin by proving Lemma~\ref{lem:sep}, which we will subsequently use to prove the general fidelity bounds stated in Theorem~\ref{thm:new}.

\lemsep*

\begin{proof}
Let us begin by recalling  the setup and introducing some notation. Bob is attempting to find a decoder
$\mcd_{V}:ER\to \widehat{M}J$  which maximises the fidelity
\begin{align*}
F_{V^B,\mcd_{V}}=F\Big(\tr_{JB_f}\left[(\mcd_{V} \ot\id_{AB_f})(\id_E\ot \mcv^B\ot \id_A)(\Phi^{EB_i}\ot \Phi^{MA})\right], \Phi^{\widehat MA} \Big)\,,
\end{align*}
where we write $\mcv^B$ for the channel which acts as  $\mcv^B(\rho)=V^B\rho (V^B)\ad$ and $J$ is a junk register. Without loss of generality, we can purify Bob's decoding operation, taking $\mcd_{V}$ to be implemented by an isometry $W:ER\to \widehat{M} Z$, followed by tracing out the ancilla space $Z$. Denoting its action on the global pure state characterising the system after the action of the black hole by
\begin{equation}
\ket{\Psi_W}:=(\id_{AB_f}\otimes W)(\id_E\ot V^B\ot \id_A)(\ket{\Phi}^{EB_i}\ot \ket{\Phi}^{MA})\,,
\end{equation}
we then have
\begin{equation}
F_{V^B,\mcd_{V}} =\bra{\Psi_W}\bigl(\Phi^{\widehat MA}\otimes \id_{B_fZ}\bigr)\ket{\Psi_W}=: \braket{v_W|v_W}^{B_fZ}\,,
\end{equation}
where we emphasise that 
\begin{equation}
 \ket{v_W}^{B_fZ}=\bigl(\bra{\Phi}^{\widehat MA}\otimes \id_{B_fZ}\bigr)\ket{\Psi_W}^{\widehat MAB_f Z}
\end{equation}
is unnormalised. Now, by Cauchy-Schwarz, $\braket{v_W|v_W}=\max_{\|\ket \eta^{B_fZ}\|=1}\lvert\braket{\eta|v_W}\rvert^2$, so that we can equivalently write
\begin{align}
   F_{\rm rec}(V^B) &=\max_{\substack{W\ad W=\id_{ER}\\ \braket{\eta|\eta}^{B_fZ}=1}}\abs{\bigl(\bra\Phi^{\widehat MA}\otimes\bra\eta^{B_fZ}\bigr)\ket{\Psi_W}^{\widehat MAB_f Z}}^2\\
   &=\max_{\substack{W\ad W=\id_{ER}\\ \braket{\eta|\eta}^{B_fZ}=1}}\abs{\bigl(\bra\Phi^{\widehat MA}\otimes\bra\eta^{B_fZ}\bigr)(\id_{AB_f}\otimes W)(\id_E\ot V^B\ot \id_A)(\ket{\Phi}^{EB_i}\ot \ket{\Phi}^{MA})}^2\,.\label{eq:uhl}
\end{align}
Now, for a fixed $\ket\eta$, let us define $ \zeta_\eta^{B_f}:=\tr_Z\ketbra{\eta}{\eta}^{B_fZ}$. Then $\ket\Phi^{\widehat MA}\otimes\ket\eta^{B_fZ}$ is a purification of $\tau^{A}\otimes\zeta_\eta^{B_f}$, so that Eq.~\eqref{eq:uhl} is an optimisation over purifications of $\sg_{V}^{AB_f}=\tr_{ER}[(\id_E\ot \mcv^B\ot \id_A)(\Phi^{EB_i}\ot \Phi^{MA})]$ of a (squared) inner product which  by Uhlmann's theorem yields
\begin{equation}
 F_{\rm rec}(V^B)= \max_{  \braket{\eta|\eta}^{B_fZ}=1}F(\sg_{V}^{AB_f}, \tau^{A}\otimes\zeta_\eta^{B_f})= \max_{  \zeta^{B_f}}F(\sg_{V}^{AB_f}, \tau^{A}\otimes\zeta^{B_f})\,,
\end{equation}
where the last equality follows from the fact that we can always find a purification of any given $\zeta^{B_f}$.

\end{proof}

\thmnew*

\begin{proof}
We again begin by recalling (and introducing) some notation. We let $\lambda_j$ be the $D:=d_A d_{B_f}$ eigenvalues of $\sigma_{V}^{AB_f}$, and put $x_j=D\lambda_j$. We write $\langle\cdot\rangle$ for the (uniform) average
over $j$. Let $  
 X_{V}^{AB_f}=\sigma_{V}^{AB_f}-\tau^A\otimes\tau^{B_f},$ and, for a specific ${V^B}\in G$, let  
\begin{equation*}
 S({V^B}):=D\big\|X_{V}^{AB_f}\big\|_2^2=D\tr[(\sigma_{V}^{AB_f})^2]-1. 
\end{equation*}
It follows immediately from these definitions that we have 
$\langle x\rangle=1$, $\langle x^2\rangle=1+S({V^B})$, and $F(\sigma_{V}^{AB_f},\tau^{A}\otimes\tau^{B_f})=\langle\sqrt{x}\rangle^2$. As Hoelder's inequality with $p=3/2$ and $q=3$  gives $\sum_j x_j \leq \big(\sum_j\sqrt{x_j}\big)^{2/3}\big(\sum_j{x_j}^2\big)^{1/3}$, we further conclude that $1=\langle x\rangle^3\leq\langle\sqrt{x}\rangle^2\langle x^2\rangle$. Taking $\zeta^{B_f}=\tau^{B_f}$ in Eq.~\eqref{eq:fid_opt} then yields
\begin{equation}
    F_{\rm rec}(V^B) \geq F\big( \sigma_{V}^{AB_f},\id_{AB_f}/D\big)=\langle\sqrt{x}\rangle^2\geq \frac{1}{\langle x^2\rangle} = \frac{1}{1+S({V^B})}\,.
\end{equation}
By taking averages and using Jensen's inequality the lower bound of Eq.~\eqref{eq:newbound} follows. \\

Establishing the other bound requires a few  technicalities. To begin,  fix an
orthonormal basis $\{\ket i\}_{i=1}^{d_A}$ of $A$, and define $P_i:=\ketbra{i}{i}\otimes \id_{B_f}$, and $\mcp_A(\sigma):=\sum_{i=1}^{d_A}P_i\sigma P_i$; by the \textit{pinching inequality}~\cite{winter2025generalised} we have $\sg_{V}^{AB_f}\leq d_A\mathcal P_A(\sg_{V}^{AB_f})$. Next, let us define $\sigma_i^{B_f}:=(\bra i\otimes \id_{B_f})
 \sg_{V}^{AB_f}(\ket i\otimes \id_{B_f})\geq0$; as $\tr_A\sg_{V}^{AB_f}$ is maximally mixed, their sum is $\sum_{i=1}^{d_A}\sigma_i^{B_f}=\tau^{B_f}$. As every term in the sum is positive, they each  obey
$\sigma_i^{B_f}\leq \tau^{B_f}= \id_{B_f}/d_{B_f}$.  Consequently,
\begin{equation}
\sg_{V}^{AB_f}\leq d_A \mathcal P_A(\sg_{V}^{AB_f})=d_A\sum_{i=1}^{d_A}\ketbra{i}{i}\otimes\sigma_i^{B_f}\leq\frac{d_A}{d_{B_f}}\id_{AB_f}\,,
\end{equation}
so that in particular every eigenvalue $\lm$ of $\sg_{V}^{AB_f}$ satisfies
$0\leq\lambda\leq {d_A}/{d_{B_f}}$.  Recalling that $x_j={d_A}{d_{B_f}}\lambda_j$
therefore gives the   bound $ 0\leq x_j\leq d_A^2$ for all $j$. Now, since $\langle x\rangle=1$,
\begin{align*}
    S(V^B)=\langle x^2\rangle-1 =\left\langle(x-1)^2\right\rangle =\left\langle(\sqrt{x}-1)^2(\sqrt{x}+1)^2\right\rangle&\leq(d_A+1)^2\left\langle(\sqrt{x}-1)^2\right\rangle\\
    &=(d_A+1)^2\left(\langle x\rangle-2\langle\sqrt{x}\rangle+1\right)\\
 &=2(d_A+1)^2(1-\langle\sqrt{x}\rangle)\\
 &=2(d_A+1)^2\Big(1-\sqrt{F(\sigma_{V}^{AB_f},\tau^{A}\otimes\tau^{B_f})}\Big)\,,
\end{align*}
where we have used $ 0\leq \sqrt{x_j}\leq d_A$ for all $j$, and that $F(\sigma_{V}^{AB_f},\tau^{A}\otimes\tau^{B_f})=\langle\sqrt{x}\rangle^2$.

Put $F_{0,V}:=F\left(\sigma_V^{AB_f},\tau^A\otimes\tau^{B_f}\right)$ and for an arbitrary state $\zeta^{B_f}$, let $F_{\zeta,V}:=F\left(\sigma_V^{AB_f},\tau^A\otimes\zeta^{B_f}\right)$. Then,
since $\sigma_V^{B_f}=\tau^{B_f}$, monotonicity of the fidelity under the partial trace gives $F(\zeta^{B_f},\tau^{B_f})\geq F_{\zeta,V}$;
using the angular triangle inequality for the fidelity,
\begin{equation}
\arccos\sqrt{F(\rho,\omega)}\leq\arccos\sqrt{F(\rho,\eta)}+\arccos\sqrt{F(\eta,\omega)},
\end{equation}
together with the multiplicativity of the fidelity under tensor products, we obtain
\begin{equation}
\arccos\sqrt{F_{0,V}}
\leq
2\arccos\sqrt{F_{\zeta,V}}.
\end{equation}
Consequently, $F_{\zeta,V}\leq (1+\sqrt{F_{0,V}})/2$, whence maximising over $\zeta^{B_f}$ and using Lemma~\ref{lem:sep} gives
\begin{equation}
    F_{\rm rec}(V^B)\leq\frac{1+\sqrt{F_{0,V}}}{2}.
\end{equation}
Combining the preceding two bounds, we obtain
\begin{align}
    1-F_{\rm rec}(V^B)&\geq\frac{1-\sqrt{F_{0,V}}}{2} \geq\frac{S(V^B)}{4(d_A+1)^2}.
\end{align}

Averaging over $V^B\sim\mu$ and using $\expect_{V^B\sim\mu}S(V^B)=d_Ad_{B_f}\Delta_2^\mu$, we finally obtain
\begin{equation}
\expect_{V^B\sim\mu}F_{\rm rec}(V^B)\leq
1-\frac{d_A d_{B_f}\Delta_2^\mu}{4(d_A+1)^2},
\end{equation}
which proves the upper bound.
\end{proof}

%----------------------------------------------------------------------

\subsection{Degree-resolved matchgate decoupling}
\label{app:polarised_mirror}

Next, we prove Theorem~\ref{thm:main}, which gives a degree-resolved expression for the decoupling error for arbitrary input states. We then evaluate the correlation weights $\mcw_k(\omega^{AB}_{\mathrm{HP}})$ for the Hayden--Preskill input state, as stated in Lemma~\ref{lem:w}.

\thmmain*

\begin{proof}\label{proof_theorem1}

Recalling our definition $\Delta_2^{\mgb,\omega}=\expect_{V^B\sim\mu_{\mgb}}\|X_{V,\omega}^{AB_f}\|_2^2$ with $X_{V,\omega}^{AB_f}=\sigma_{V,\omega}^{AB_f}-\omega^A\otimes\sigma_{V,\omega}^{B_f}$, we first see that the decoupling error can be expressed in terms of $D^{AB}:=\omega^{AB}-\omega^A\otimes\omega^B$:
\begin{align}
    \Delta_2^{\mgb,\omega}&=\expect_{V^B\sim\mu_{\mgb}}\|X_{V,\omega}^{AB_f}\|_2^2\notag\\
    &=\expect_{V^B\sim\mu_{\mgb}}\tr\Big[\left(X_{V,\omega}^{AB_f}\right)^2\Big]\\
    &=\expect_{V^B\sim\mu_{\mgb}}
        \tr\left[\left(\tr_R[(\id_A\otimes V^B)D^{AB}(\id_A\otimes (V^B)^\dagger)]\right)^2
    \right]\notag
\end{align}
Here, we used that $X_{V,\omega}^{AB_f}$ is Hermitian. Using the swap trick~\cite{mele2023introduction} and leaving identity operators on the two copies of $A$ implicit, we obtain the following two-copy representation:
\begin{align}
    \expect_{V^B\sim\mu_{\mgb}}
        \tr\Big[\left(X_{V,\omega}^{AB_f}\right)^2\Big]
    &=\expect_{V^B\sim\mu_{\mgb}}
        \tr\left[\left(\tr_R[(\id_A\otimes V^B)D^{AB}(\id_A\otimes (V^B)^\dagger)]\right)^2\right]
    \notag\\
    &=\expect_{V^B\sim\mu_{\mgb}}\tr\img[2.0cm]{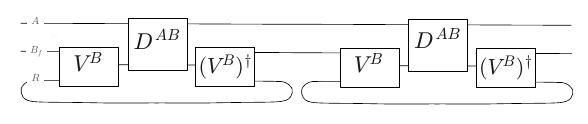}\notag\\
    &=\expect_{V^B\sim\mu_{\mgb}}\tr\img[3.5cm]{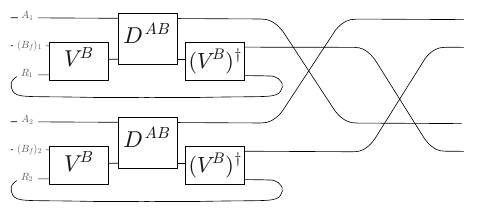}\notag\\
    &=\expect_{V^B\sim\mu_{\mgb}}
        \tr\left[(V^B)^{\otimes 2}(D^{AB})^{\otimes 2}((V^B)^{\dagger})^{\otimes 2} \mbs_A \widetilde \mbs_{B_f}\right]\notag\\
&=\tr\left[(D^{AB})^{\otimes2}\left(\mbs_A\otimes\expect_{V^B\sim\mu_{\mgb}}\left[((V^B)^\dagger)^{\otimes2}\widetilde{\mbs}_{B_f}(V^B)^{\otimes2}\right]\right)\right].\label{eq:swap_arbitrary_state}
\end{align}
This is the same swap-trick representation as explained in the preliminaries in Figure~\ref{fig:purity} and Eq.~\eqref{eq:pure_swap}, with the correlation operator $D^{AB}$ replacing the Hayden--Preskill input-specific expression. Here, $\mbs_A$ is the swap operator on the two copies of $A$, while
\begin{equation}
    \widetilde{\mbs}_{B_f}:=\mbs_{B_f}\otimes \id_{R_1R_2}=\frac{1}{d_{B_f}}\sum_{S\subseteq[2n_{B_f}]}\widetilde\gamma_S\otimes\widetilde\gamma_S,\qquad\widetilde\gamma_S:=\gamma_S^{B_f}\otimes \id_R
\end{equation}
is the swap operator on the two copies of $B_f$ extended by the identity on the radiation copies.

The two-copy matchgate twirl is the Hilbert--Schmidt orthogonal projection onto $\operatorname{Comm}\{V^{\otimes 2}: V\in\mgb\}$. Using the orthonormal basis introduced in Eqs.~\eqref{eq:comm_basis_0} and~\eqref{eq:comm_basis_1} for $n\equiv n_B$, the projection formula gives
\begin{align}
    \expect_{V^B\sim\mu_{\mgb}}\left[ ((V^B)^{\dagger})^{\otimes 2}\widetilde{\mbs}_{B_f}(V^B)^{\otimes 2}\right]=\sum_{k=0}^{2n_{B}}\left(\tr[Q_{k,0}\widetilde{\mbs}_{B_f}]Q_{k,0}+\tr[Q_{k,1}\widetilde{\mbs}_{B_f}] Q_{k,1}\right).\label{eq:projection_formula}
\end{align}
We start by calculating the scalar coefficient
\begin{align}
    \tr[Q_{k,0}\widetilde{\mbs}_{B_f}]&=\frac{1}{d_{B_f}d_B\sqrt{\binom{2n_B}{k}}}\sum_{\substack{P\in\mcb_k^{(n_B)}\\S\subseteq[2n_{B_f}]}}\tr[(P\otimes P)(\widetilde \gamma_S\otimes \widetilde \gamma_S)]\notag\\
    &=\frac{1}{d_{B_f}d_B\sqrt{\binom{2n_B}{k}}}\sum_{\substack{P\in\mcb_k^{(n_B)}\\S\subseteq[2n_{B_f}]}}(\tr[P\widetilde \gamma_S ])^2\notag\\
   &=\frac{d_B}{d_{B_f}}\frac{\binom{2n_{B_f}}{k}}{\sqrt{\binom{2n_B}{k}}},
\end{align}
where we used $\tr[P\widetilde \gamma_S ]=d_B\delta_{P,\widetilde\gamma_S}$ which allows $\binom{2n_{B_f}}{k}$ possibilities. 
Multiplying by $Q_{k,0}$, see Eq.~\eqref{eq:comm_basis_0}, leads to
\begin{align}
    \tr[Q_{k,0}\widetilde{\mbs}_{B_f}]Q_{k,0}=\frac{1}{d_{B_f}}\frac{\binom{2n_{B_f}}{k}}{\binom{2n_B}{k}}\sum_{P\in\mcb_k^{(n_B)}} P\otimes P.
\end{align}
By contrast, the parity-twisted sector does not contribute:
\begin{align}
\tr[Q_{k,1}\widetilde{\mbs}_{B_f}]&=\frac{i^{k\bmod 2}}{d_{B_f}d_B\sqrt{\binom{2n_B}{k}}}\sum_{\substack{P\in\mcb_k^{(n_B)}\\S\subseteq[2n_{B_f}]}}\tr\left[(P\otimes\Pi_{n_B} P)(\widetilde \gamma_S \otimes \widetilde \gamma_S)\right]\\
&=\frac{i^{k\bmod 2}}{d_{B_f}d_B\sqrt{\binom{2n_B}{k}}}\sum_{\substack{P\in\mcb_k^{(n_B)}\\S\subseteq[2n_{B_f}]}}\tr\left[P\widetilde\gamma_S\right]\tr\left[\Pi_{n_B}P\widetilde\gamma_S\right]\\
&=0.
\end{align}
Indeed, the first Majorana trace is nonzero only if $P=\widetilde\gamma_S$. In this case, the second trace equals
$\tr_B(\Pi_{n_B}P\widetilde\gamma_S)=\tr_B(\Pi_{n_B}P^2)=\tr_B(\Pi_{n_B})=0$.
By these results, we obtain
\begin{equation}
\expect_{V^B\sim\mu_{\mgb}}\left[((V^B)^\dagger)^{\otimes2}\widetilde{\mbs}_{B_f}(V^B)^{\otimes2}\right]=\frac{1}{d_{B_f}}\sum_{k=0}^{2n_{B_f}}\frac{\binom{2n_{B_f}}{k}}{\binom{2n_B}{k}}\sum_{P\in\mcb_k^{(n_B)}}P\otimes P.
\label{eq:partial_swap_twirl}
\end{equation}
Here, we used that all terms with $k>2n_{B_f}$ vanish in Eq.~\eqref{eq:projection_formula}, since $\binom{2n_{B_f}}{k}=0$ for $k>2n_{B_f}$.
To make it easier to assemble the final result for 
$\Delta_2^{\mgb,\omega}$, we introduce for $P\in\mcb_k^{(n_B)}$, the operator
\begin{equation}
    D_P^A:=\tr_B\left[D^{AB}(\id_A\otimes P)\right].
\end{equation}
Using the cyclicity of the partial trace with respect to operators acting only on $B$, it immediately follows that $D_P^A$ is Hermitian. Looking at the structure of Eq.~\eqref{eq:swap_arbitrary_state} and the structure of our matchgate twirl result in Eq.~\eqref{eq:partial_swap_twirl}, a second application of the swap trick gives
\begin{align}
    \tr\left[(D^{AB})^{\otimes2}\left(\mbs_A\otimes P_{B_1}\otimes P_{B_2}\right)\right]
    &=\tr\left[(D_P^A\otimes D_P^A)\mbs_A\right]=\tr[(D_P^A)^2]=\|D_P^A\|_2^2.
\end{align}
Thus, we finally obtain
\begin{align*}
\expect_{V^B\sim\mu_{\mgb}}\left[
\tr\left[(X^{AB_f}_{V,\omega})^2\right]\right]
&=\frac{1}{d_{B_f}}\sum_{k=0}^{2n_{B_f}}\frac{\binom{2n_{B_f}}{k}}{\binom{2n_B}{k}}\sum_{P\in \mcb_k^{(n_B)}}\tr\left[(D^{AB})^{\otimes2}\left(\mbs_A\otimes P_{B_1}\otimes P_{B_2}\right)\right]\\
&=\frac{1}{d_{B_f}}\sum_{k=0}^{2n_{B_f}}\frac{\binom{2n_{B_f}}{k}}{\binom{2n_B}{k}}\sum_{P\in \mcb_k^{(n_B)}}\tr\left[\left(\tr_B\left[D^{AB}(\id_A\otimes P)\right]\right)^2\right]\\
&=\frac{1}{d_{B_f}}\sum_{k=1}^{2n_{B_f}}\frac{\binom{2n_{B_f}}{k}}{\binom{2n_B}{k}}\sum_{P\in \mcb_k^{(n_B)}}\left\|\left(\tr_B\left[D^{AB}(\id_A\otimes P)\right]\right)\right\|_2^2
\end{align*}
The degree-zero contribution vanishes because $\mcb_0=\{\id_B\}$ and
\begin{equation}
    D^A=\tr_B(D^{AB})=\omega^A-\omega^A\tr(\omega^B)=0.
\end{equation}
Using the definitions of $p_k(n_{B_f},n_B)$~\eqref{eq:p} and $\mcw_k(\omega)$~\eqref{eq:wdef}, the result follows:
\begin{equation}
    \Delta_2^{\mgb,\omega}=\frac{1}{d_{B_f}}\sum_{k=1}^{2n_{B_f}}
    p_k(n_{B_f},n_B)\,\mcw_k(\omega).
\end{equation}
\end{proof}

%-------------------------------------------------

\lemw*
\begin{proof}
Recall that $\omega_{\mathrm{HP}}^{AB}=\Phi^{MA}\otimes\tau^{B_i}$, and let $d_A=d_M=2^m$. The marginals of $\omega_{\mathrm{HP}}^{AB}$ are 
\begin{align}
\omega^A&=\tr_B\left[\omega^{AB}_{\mathrm{HP}}\right]=\tr_M\left[\Phi^{MA}\right]\tr_{B_i}\left[\tau^{B_i}\right]=\frac{1}{d_A}\id_A\\
\omega^B&=\tr_A\left[\omega^{AB}_{\mathrm{HP}}\right]=\tr_A\left[\Phi^{MA}\right]\otimes\tau^{B_i}=\frac{1}{d_M}\id_{M}\otimes\tau^{B_i}.
\end{align}
Hence, the corresponding correlation operator is
\begin{align}
    D^{AB}_{\mathrm{HP}}=\left(\Phi^{MA}-\frac{1}{d_M d_A}\id_{MA}\right)\otimes\tau^{B_i}.
\end{align}
Recall that, in the input Jordan--Wigner ordering fixed in this work, the modes of $M$ come first, followed by the modes of $B_i$.

Since $B_i$ is maximally mixed and the modes of $M$ precede those of $B_i$, tracing over $B_i$ annihilates every Majorana monomial containing at least one Majorana generator from $B_i$. Hence, only Majorana monomials supported entirely on $M$ contribute. These are precisely $P=P_M\otimes \id_{B_i}$ with $P_M\in\mcb_k^{(m)}$.
For $k\geq 1$, we have $P_M\neq \id_{M}$ and hence $\tr(P_M)=0$. Therefore,
\begin{align}
\tr_B\left[D_{\mathrm{HP}}^{AB}(\id_A\otimes P)
    \right]&=\tr_M\left[\left(\Phi^{MA}-\frac{\id_{MA}}{d_M d_A}\right)(\id_A\otimes P_M)
    \right]\notag\\
    &=\tr_M\left[\Phi^{MA}(\id_A\otimes P_M)\right].
\end{align}
Using the maximally entangled-state identity, 
\begin{equation}
    \tr_M\left[\Phi^{MA}(P_M\otimes\id_A)\right]=\frac{P_A^{\mathsf T}}{d_M},
\end{equation}
and since every Majorana monomial
$P_M\in\mcb_k^{(m)}$ satisfies $P_M^2=\id_{M}$,
the weight of each contributing monomial is
\begin{align}
    \left\|\tr_M\left[\Phi^{MA}( P_M\otimes \id_A)\right]\right\|_2^2
    &=\frac{1}{d_M^2}\tr(P_A^2)=\frac{1}{d_M}.
\end{align}
There are $\binom{2m}{k}$ degree-$k$ Majorana monomials on $M$. Consequently,
\begin{align}
    \mcw_k(\omega_{\mathrm{HP}}^{AB})
    &=\sum_{P_M\in\mcb_k^{(m)}}\left\|\tr_M[\Phi^{MA}( P_M\otimes \id_A)]\right\|_2^2=\sum_{P_M\in\mcb_k^{(m)}}\frac{1}{d_M}=\frac{1}{d_M}\binom{2m}{k}.
\end{align}
For $k>2m$, no degree-$k$ Majorana monomial can be supported entirely on $M$. Hence, $\mcw_k(\omega_{\mathrm{HP}}^{AB})=0$ for $k>2m$.

\end{proof}

%----------------------------------------------------------------------

\subsection{Obstruction to information recovery}
\label{app:no_recovery}

 We now prove the bounds on the average optimal recovery fidelity under matchgate dynamics stated in Theorem~\ref{thm:main_3}. The lower bound follows directly from the results of the previous subsection and Theorem~\ref{thm:new}. To establish the upper bound, we first prove Lemma~\ref{lem:mg_lower_bound}, which bounds the decoupling error from below. All results are now specialised to the Hayden--Preskill input state $\omega_{\rm HP}^{AB}$.

\lemhp*

\begin{proof}
%setup idea
We use the duality between the trace norm and the operator norm, which implies that, for every Hermitian observable $Y_{V}$ satisfying $\|Y_{V}\|_{\infty}\le 1$,
\begin{align}
    \|X_{V}^{AB_f}\|_1\ge\big\lvert\tr\big[Y_{V}X_{V}^{AB_f}\big]\big\rvert.
\end{align}
As
\begin{align}
\tr(Y_{V}X_{V}^{AB_f})=\tr(Y_{V}\sigma_{V}^{AB_f})-\tr(Y_{V}(\tau^A\otimes\tau^{B_f})),
\end{align}
$Y_{V}$ is acting as a sort of correlation witness: its overlap with $X_{V}^{AB_f}$ measures the difference between its expectation value in the actual output state and in the corresponding marginalised state. 

To obtain a strong lower bound, we seek a witness with a large overlap with $X_{V}^{AB_f}$. The degree-dependent survival factors appearing in Theorem~\ref{thm:main} are largest in the degree-one sector, which motivates constructing a witness that detects a degree-one Majorana correlation. The construction is possible for the Hayden--Preskill input state because the maximally entangled state contains perfectly normalised correlations between each message Majorana and its partner on $A$:
\begin{align}
   \tr\left[\omega_{\mathrm{HP}}^{AB}(\gamma_1^A\otimes\gamma_1^M\otimes \id^{B_i})\right]=1.
\end{align}
Using the convention from Eq.~\eqref{eq:majo_linearly}, we obtain
\begin{align}
    V^B\gamma_1^M(V^B)^{\dagger}&=\sum_{\nu=1}^{2n_B}O_{1\nu}(V^B)\gamma_\nu^B\\
    &=\sum_{\nu=1}^{2n_{B_f}}O_{1\nu}(V^B)\gamma_\nu^{B_f}+\sum_{\nu=2n_{B_f}+1}^{2n_B}O_{1\nu}(V^B)\gamma_\nu^B.
\end{align}
Define the component of the transformed Majorana supported on $B_f$ by
\begin{align}
    \Gamma_{V,f}=\sum_{\nu=1}^{2n_{B_f}}O_{1\nu}(V^B)\gamma_\nu^{B_f}.
\end{align}
We further define its length by
\begin{align}\label{eq:sv2}
    s_{V}^2:=\sum_{\nu=1}^{2n_{B_f}}O_{1\nu}(V^B)^2,
\end{align}
so that $\Gamma_{V,f}^2=s_{V}^2 \id_{B_f}$,
and consequently $\|\Gamma_{V,f}\|_{\infty}=s_{V}$.
Moreover, since $O(V^B)$ is orthogonal, $\sum_{\nu=1}^{2n_B}O_{1\nu}(V^B)^2=1$, and therefore $0\le s_{V}\le 1$. If $s_V=0$, the desired pointwise inequality $\|X_V^{AB_f}\|_1^2\geq s_V^2$ is trivial. We may therefore assume
that $s_V>0$. Define the normalised operator $\hat\Gamma_{V,f}:={\Gamma_{V,f}}/{s_{V}}$, and choose  the correlation witness $ Y_{V}:=\gamma_1^A\otimes \hat\Gamma_{V,f},$
with $\|Y_{V}\|_{\infty}=\|\gamma_1^A\|_{\infty}\|\hat\Gamma_{V,f}\|_{\infty}=1$.
We now evaluate the overlap of $Y_{V}$ with the output correlation operator:
\begin{align}
    \tr\left[Y_{V}X_{V}^{AB_f}\right]=\tr\left[\left(\gamma_1^A\otimes\hat\Gamma_{V,f}\right)\sigma_{V}^{AB_f}\right]-\tr\left[\left(\gamma_1^A\otimes\hat\Gamma_{V,f}\right)\left(\tau^A\otimes\tau^{B_f}\right)\right].
\end{align}
For the first summand, we obtain
\begin{align}
\tr\left[\left(\gamma_1^A\otimes\hat\Gamma_{V,f}\right)\sigma_{V}^{AB_f}\right]&=\frac{1}{s_{V}}\tr\left[\left(\gamma_1^A\otimes\Gamma_{V,f}\right)\tr_R\left[\left(\id_A\otimes V^B\right)\omega^{AB}_{\mathrm{HP}}\left(\id_A\otimes (V^B)^\dagger\right)\right]\right]\notag\\
&=\frac{1}{s_{V}d_{B_i}}\tr\left[\left(\gamma_1^A\otimes\Gamma_{V,f}\right)\left(\id_A\otimes V^B\right)(\Phi^{MA}\otimes\id_{B_i})\left(\id_A\otimes (V^B)^\dagger\right)\right]\\
&=\frac{1}{s_{V}d_{B_i}}\tr\img[2.0cm]{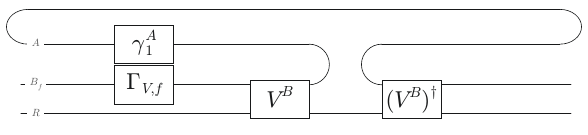}\notag\\
 &=\frac{1}{s_{V}d_B}\tr\img[1.8cm]{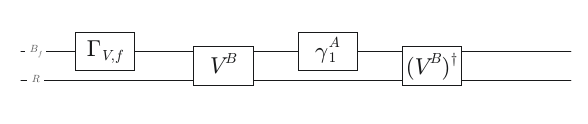}\notag\\
    &=\frac{1}{s_{V}d_{B}}\tr_B\left[\Gamma_{V,f}V^B\gamma_1^{M}(V^B)^\dagger\right]\notag\\
    &=\frac{d_R}{s_{V}d_{B}}\tr_{B_f}\left[\Gamma_{V,f}^2\right]\notag\\
    &=\frac{s_{V}}{d_{B_f}}\tr_{B_f}\left[\id^{B_f}\right]\\
    &=s_{V}.
\end{align}
The second summand vanishes,
\begin{align}
\tr\left[\left(\gamma_1^A\otimes\hat\Gamma_{V,f}\right)\left(\omega^A\otimes\sigma_{V}^{B_f}\right)\right]&=\tr\left[\gamma_1^A\omega^A\right]\tr\left[\hat\Gamma_{V,f}\sigma_{V}^{B_f}\right]=\frac{1}{d_A}\tr\left[\gamma_1^A\right]\tr\left[\hat\Gamma_{V,f}\sigma_{V}^{B_f}\right]= 0.
\end{align}
Thus, we obtain $\tr\big[Y_{V}X_{V}^{AB_f}\big]=s_{V}$,
and consequently, $\|X_{V}^{AB_f}\|^2_1\ge s_{V}^2$; averaging over $V^B\in\mgb$ yields
\begin{align}
  \expect_{V\sim\mu_{\mgb} } \|X_{V}^{AB_f}\|^2_1  \ge\expect_{V\sim\mu_{\mgb} } s_{V}^2
\end{align}
and it remains to evaluate $\expect_{V^B\sim\mu_{\mgb}} s_{V}^2$. Now, the element $O\in\mbso(2n_B)$ associated to a Haar-random matchgate $V^B$ is itself Haar randomly distributed on $\mbso(2n_B)$, and in particular satisfies $\expect_{V^B\sim\mu_{\mgb}}\left[O_{\mu\nu}(V^B)^2\right]={1}/{2n_B}$ for all $\mu,\nu$.

It follows immediately from Eq.~\eqref{eq:sv2} that $\expect_{V^B\sim\mu_{\mgb}} s_{V}^2={n_{B_f}}/{n_B}$; since $n_{B_f}=n_B-r$, we conclude $\Delta_1^{\mgb}\ge 1-{r}/{n_B}$, as desired.

\end{proof}

%----------------------------------------------

\thmmgfidelity*

\begin{proof}
Using the correlation weights $\mcw_k(\omega^{AB}_{\mathrm{HP}})$ from Lemma~\ref{lem:w}, and since $d_A=d_M$, we obtain $d_{B_f}d_M\Delta_2^{\mgb}=R_{\mathrm{MG}}$; 
substituting this expression into the lower bound from Theorem~\ref{thm:new} yields the claimed lower bound. 
We now prove the upper bound. For each $V^B$ with $s_{V}>0$, let
\begin{align}
    Y_{V}=\gamma_1^A\otimes\widehat{\Gamma}_{V,f}
\end{align}
be the correlation witness constructed in the proof of Lemma~\ref{lem:mg_lower_bound}, which we recall to satisfy $Y_{V}^2=\id_{AB_f}$ and 
\begin{align}
     \tr\,[Y_{V}\sigma_{V}^{AB_f}]=s_{V}.\label{eq:witness_expectation}
\end{align}
Moreover, for every state $\zeta^{B_f}$,
\begin{align}
\tr\left[Y_{V}\left(\tau^A\otimes\zeta^{B_f}\right)\right]&=\tr[\gamma_1^A\tau^A]\,\tr[\widehat{\Gamma}_{V,f}\zeta^{B_f}]=0,\label{eq:witness_product_expectation}
\end{align}
because $\tau^A=\id_A/d_A$ and $\gamma_1^A$ is traceless.
Consider the two-outcome projective measurement
\begin{align}
    P_{V,\pm}:=\frac{1}{2}\left(\id_{AB_f}\pm Y_{V}\right).
\end{align}
Eqs.~\eqref{eq:witness_expectation} and~\eqref{eq:witness_product_expectation} imply that its outcome probabilities on $\sigma_{V}^{AB_f}$ and $\tau^A\otimes\zeta^{B_f}$ are, respectively,
\begin{align}
    p_{V,\pm}&=\tr[P_{V,\pm}\sigma_{V}^{AB_f}]=\frac{1\pm s_{V}}{2},\qquad q_{V,\pm}=\tr\left[P_{V,\pm}\left(\tau^A\otimes\zeta^{B_f}\right)\right]=\frac{1}{2}.
\end{align}
For two classical probability distributions $p$ and $q$, our fidelity convention gives $F(p,q)=\left(\sum_{\alpha}\sqrt{p_\alpha q_\alpha}\right)^2$.

Applying the measurement channel $\mathcal M_V(\rho):=\sum_{\alpha\in\{+,-\}}\tr[P_{V,\alpha}\rho]\ket{\alpha}\!\bra{\alpha}$ therefore gives
\begin{align}
    F\!\left(\sigma_{V}^{AB_f},\tau^A\otimes\zeta^{B_f}\right)&\leq\left(\sqrt{p_{V,+}q_{V,+}}+\sqrt{p_{V,-}q_{V,-}}\right)^2\\
    &=\frac{1}{4}\left(\sqrt{1+s_{V}}+\sqrt{1-s_{V}}\right)^2\\
    &=\frac{1}{2}\left(1+\sqrt{1-s_{V}^2}\right).
\end{align}
Since this bound holds for every $\zeta^{B_f}$, Lemma~\ref{lem:sep} implies
\begin{align}
    F_{\mathrm{rec}}(V^B)\leq\frac{1}{2}\left(1+\sqrt{1-s_{V}^2}\right).
    \label{eq:pointwise_recovery_upper_bound}
\end{align}
For $s_{V}=0$, Eq.~\eqref{eq:pointwise_recovery_upper_bound} follows directly from $F_{\mathrm{rec}}(V^B)\leq 1$, so it holds for every $V^B$.
Averaging and applying Jensen's inequality to the concave function $x\mapsto\sqrt{1-x}$ gives
\begin{align}
    \expect_{V^B\sim\mu_{\mgb}} F_{\mathrm{rec}}(V^B)&\leq\frac{1}{2}\left(1+\expect_{V^B\sim\mu_{\mgb}}\sqrt{1-s_{V}^2}\right)\\
    &\leq\frac{1}{2}\left(1+\sqrt{1-\expect_{V^B\sim\mu_{\mgb}}s_{V}^2}\right)\\
    &= \frac{1}{2}\left(1+\sqrt{\frac{r}{n_B}}\right),
\end{align}
where we have used that, as shown in the proof of Lemma~\ref{lem:mg_lower_bound}, $\expect_{V^B\sim\mu_{\mgb}}s_{V}^2=1-{r}/{n_B}$. 
This proves the claim.
\end{proof}

%----------------------------------------------
\subsection{Recovery with a non-Gaussian horizon}
\label{app:non_gaussian}

Finally, we prove Theorem~\ref{thm:augrec} by combining the Clifford construction from Section~\ref{sec:nongaussianhorizon} with the degree-resolved decoupling formula of Theorem~\ref{thm:main}.

\thmaugrec*

\begin{proof}
We are going to apply Theorem~\ref{thm:main} to $\omega_{C_{t^*}}=(\id_A\otimes C_{t^*})\omega^{AB}_{\mathrm{HP}}(\id_A\otimes C_{t^*}^\dagger)$. Since the total dynamics are given by $V^B C_{t^*}$ with $V^B\sim\mu_{\mgb}$, we have $\Delta_2^{\mgb C_{t^*}}=\Delta_2^{\mgb,\omega_{C_{t^*}}}$. Moreover, since $\omega_{\mathrm{HP}}^B=\tau^B$ and $C_{t^*}\tau^BC_{t^*}^\dagger=\tau^B$, the correlation operator of the transformed input is
\begin{equation}
    D_{C_{t^*}}^{AB}=(\id_A\otimes C_{t^*})D_{\mathrm{HP}}^{AB}(\id_A\otimes C_{t^*}^\dagger).
\end{equation}
For each nonzero message label $x\in\mathbb F_2^{2m}$, we define $P^M_x:=\gamma_x^M$, and  write $C_{t^*}(P_x^M\otimes I_{B_i})C_{t^*}^{\dagger}=\pm Q_x^B$, where
each $Q_x^B$ is a Majorana monomial of degree $|C_{t^*}(x)|=t^*4^m+|x|$.
Its correlation weight remains $1/d_A$: the corresponding partial-trace coefficient on $A$ is, up to sign, $P_x^{\mathsf T}/d_A$. There are $\binom{2m}{j}$ labels with $|x|=j$. Consequently, $\mcw_{t^*4^m+j}(\omega_{C_{t^*}})
=d_A^{-1}\binom{2m}{j},$ for all $1\le j\le 2m$, and all other weights vanish. Thus, Theorem~\ref{thm:main} yields
\begin{equation*}
    d_A d_{B_f}\Delta_2^{\mgb C_{t^*}}=\sum_{j=1}^{2m}\binom{2m}{j}\frac{\binom{2n_B-2r}{t^*4^m+j}}
     {\binom{2n_B}{t^*4^m+j}}.
\end{equation*}
Now let $k_j=t^*4^m+j$. 
%Since $t^*=(n_B-m)/(4^m-1)+\mco(1)$, we have $k_j/(2n_B)=1-q_m^{-1}+O_m(n_B^{-1})$. 
Since $t^*=(n_B-m)/(4^m-1)+\mco(1)$, we have,   for $1\leq j\leq 2m$,
\begin{align}
    1-\frac{k_j}{2n_B}&=1-\frac{4^m}{2(4^m-1)}+\mco_m(n_B^{-1})\notag\\
    &=\frac 12-\frac{1}{2(4^m-1)}+\mco_m(n_B^{-1})\notag\\
    &= q_m^{-1}+\mco_m(n_B^{-1}),
\end{align}
where $q_m:=\frac{2(4^m-1)}{4^m-2}$.
For fixed $r$, we then rewrite each ratio as
\begin{equation*}
p_{k_j}=\frac{\binom{2n_B-2r}{k_j}}{\binom{2n_B}{k_j}}=\frac{\binom{2n_B-k_j}{2r}}{\binom{2n_B}{2r}}=\prod_{\ell=0}^{2r-1}\frac{2n_B-k_j-\ell}{2n_B-\ell}=q_m^{-2r}+O_{m,r}(n_B^{-1}).
\end{equation*}
Summing and using $\sum_{j=1}^{2m}\binom{2m}{j}=4^m-1$ then proves the claim.
\end{proof}

\end{document}